\documentclass[a4paper,UKenglish,cleveref,autoref,thm-restate]{lipics-v2021}

\usepackage[T1]{fontenc}
\usepackage[utf8]{inputenc}
\usepackage{lmodern}
\usepackage{amsmath,amssymb}
\usepackage{microtype}
\usepackage{xspace}
\usepackage[ruled,vlined]{algorithm2e}
\usepackage{flafter}
\usepackage{placeins}

\title{QPTAS for Max Weight Independent Set in Graphs with Strongly Edge-Sublinear Separators}
\titlerunning{QPTAS for MWIS in Graphs with Strongly Edge-Sublinear Separators}

\author{\'{E}douard Bonnet}{CNRS, ENS de Lyon, Université Claude Bernard Lyon 1, LIP UMR 5668, Lyon, France \and \url{http://perso.ens-lyon.fr/edouard.bonnet}}{edouard.bonnet@ens-lyon.fr}{https://orcid.org/0000-0002-1653-5822}{}
\author{Yeonsu Chang}{Department of Mathematics, Hanyang University, Seoul, South Korea \and \url{https://yeonsuchang.com/}}{yeonsu@hanyang.ac.kr}{}{}

\authorrunning{\'E. Bonnet and Y. Chang}

\Copyright{Édouard Bonnet, Yeonsu Chang}

\nolinenumbers
\hideLIPIcs
\makeatletter

\renewcommand{\endabstract}{%
  \par
  \protected@write\@auxout{}{\string\gdef\string\@pageNumberEndAbstract{\thepage}}}
\AtBeginDocument{%
  \hypersetup{%
    pdfsubject={Deterministic quasipolynomial approximation of maximum weight independent set from edge-sublinear separators},
    pdfkeywords={Independent set, approximation scheme, graph separators, induced minors}}}
\makeatother

\newcommand{\mis}{\textnormal{\textsc{Max Independent Set}}\xspace}
\newcommand{\wmis}{\textnormal{\textsc{Max Weight Independent Set}}\xspace}
\newcommand{\swmis}{\textnormal{\textsc{MWIS}}\xspace}
\newcommand{\smis}{\textnormal{\textsc{MIS}}\xspace}
\newcommand{\Oh}{\mathcal{O}}
\newcommand{\tw}{\operatorname{tw}}
\newcommand{\sn}{\operatorname{sn}}
\newcommand{\conv}{\operatorname{conv}}
\newcommand{\E}{\mathbb{E}}
\newcommand{\Solve}{\text{\normalfont\textsc{Solve}}\xspace}
\newcommand{\RoundOrCut}{\text{\normalfont\textsc{RoundOrCut}}\xspace}
\renewcommand{\leq}{\leqslant}
\renewcommand{\geq}{\geqslant}

\SetKw{Return}{return}
\DontPrintSemicolon
\begin{document}
\raggedbottom
\maketitle

\begin{abstract}
  A~hereditary graph class has strongly edge-sublinear separators if, for some constant \mbox{$0\leq b<1$}, its \mbox{$m$-edge} graphs admit balanced separators of size $\Oh(m^b)$.
  We show that \wmis admits a~deterministic quasipolynomial-time approximation scheme (QPTAS) on every such class.
  In particular, this holds on graphs excluding any fixed graph as an induced minor, by a separator theorem of \hyperlink{cite.Korhonen24}{[Korhonen and Lokshtanov, SODA~'24]}.

  Our theorem generalizes the QPTASes for axis-parallel rectangles \hyperlink{cite.Adamaszek13}{[Adamaszek and Wiese, FOCS~'13]} and polygons \hyperlink{cite.Adamaszek19}{[Adamaszek, Har-Peled, and Wiese, JACM~'19]}, or equivalently for packing connected subgraphs of a~planar graph \hyperlink{cite.Pilipczuk18}{[Pilipczuk, van Leeuwen, and Wiese, Algorithmica~'20]}.
  It also generalizes the QPTAS for \swmis on hereditary classes with strongly sublinear separators measured in vertices [\hyperlink{cite.Dvorak18}{Dvo\v{r}\'ak, EJC~'18}; \hyperlink{cite.Dvorak21}{Dvo\v{r}\'ak and Lahiri, ESA~'21}], as well as the QPTASes of \hyperlink{cite.Chudnovsky25}{[Chudnovsky et~al., arXiv~'25]}, \hyperlink{cite.Bonnet26d}{[Bonnet et~al., SWAT~'26]}, and \hyperlink{cite.Czyzewska26}{[Czy\.{z}ewska et~al., arXiv~'26]} for graphs excluding specific induced minors.
\end{abstract}

\section{Introduction}\label{sec:intro}

\mis (\smis) asks for a~largest set of pairwise nonadjacent vertices in a~graph~$G$.
This problem belongs to Karp's list of NP-hard problems~\cite{Karp72}.
Furthermore, H{\aa}stad's hardness result~\cite{Hastad99}, derandomized by Zuckerman~\cite{Zuckerman07}, shows that, for every fixed $\varepsilon>0$, approximating \smis within a~factor~$n^{1-\varepsilon}$ on $n$-vertex graphs is NP-hard.
This makes \smis a~particularly inapproximable problem.

However, in some special classes of graphs, \smis becomes (much more) tractable.
For instance, it admits a~polynomial-time exact algorithm on graphs of bounded treewidth~\cite{Arnborg89}, graphs of bounded clique-width~\cite{Courcelle00,Oum06}, $P_6$-free graphs~\cite{Grzesik22}, and circle graphs~\cite{Gavril73};
a~poly\-nomial-time approximation scheme (PTAS) on planar graphs~\cite{Baker94} and, more generally, classes excluding a~fixed minor~\cite{Grohe03}, on pseudo-disk intersection graphs~\cite{Chan12}, and on represented disk graphs with arbitrary radii~\cite{Erlebach05};
an~exact quasipolynomial-time algorithm on $P_t$-free graphs for every fixed~$t$~\cite{Gartland20} and on graphs with no induced cycle of length at least~$t$ for every fixed $t\geq4$~\cite{Gartland21};
and a~quasipolynomial-time approximation scheme (QPTAS) on represented intersection graphs of axis-parallel rectangles~\cite{Adamaszek13}, polygons~\cite{Adamaszek19}, and equivalently connected subgraphs of a~planar graph~\cite{Pilipczuk18}.

Four years ago, Dallard, Milani\v{c}, and \v{S}torgel~\cite[Question 8.2]{Dallard24c} asked if \smis admits a~polynomial algorithm on every class excluding a~fixed planar graph as an induced minor (i.e., a~graph obtained via vertex deletions and edge contractions).
A~positive solution to this question would have been a~formidable breakthrough: it would have unified, extended, and improved the various algorithms developed on the classes above with a known polynomial or quasipolynomial algorithm (indeed, these classes exclude some planar induced minors).
Independently, Gartland and Lokshtanov conjectured the same~\cite[Conjecture~1.4.2]{Gartland23}, or at~least that a~quasipolynomial-time algorithm exists.

Recently, the authors~\cite{Bonnet26c} negatively resolved Question 8.2 of~\cite{Dallard24c} (and refuted Gartland and Lokshtanov's conjecture) by showing that \smis remains NP-complete in graphs excluding the $5 \times 5$ grid as an induced minor; ruling out, unless NP $\subseteq$ QP, even a~quasipolynomial algorithm in this special case.
This reduction, however, does not rule out approximation schemes.

The existence of a~QPTAS is discussed in~\cite[Section~1.4]{Gartland23}.
A~proposed route is through the so-called \emph{structural induced grid minor conjecture} of Gartland and Lokshtanov~\cite[Conjecture~1.4.1]{Gartland23}.
This is the conjecture that, for every planar graph~$H$ (or equivalently every $k \times k$ grid), every graph excluding~$H$ as an induced minor has a~balanced separator in the neighborhood of $\Oh_H(1)$ vertices.
It is indeed known that proving this conjecture for some graph~$H$ yields a~QPTAS for \smis on $H$-induced-minor-free graphs~\cite[Theorem~1.2.6]{Gartland23}.

\begin{conjecture}[Gartland--Lokshtanov]\label{conj:planar}
  For every planar graph $H$, \mis admits a~QPTAS in $H$-induced-minor-free graphs. 
\end{conjecture}

This conjecture has been repeated in several papers, for example~\cite[Conjecture~7]{Bonnet26c} or \cite{Nikabadi26}, as a~challenging open question.

\subparagraph*{Our result.}
We positively resolve \cref{conj:planar} in much greater generality.
We say that a~hereditary graph class $\mathcal C$ has \emph{strongly edge-sublinear separators} (or equivalently, \emph{strongly edge-sublinear treewidth}) if there are constants~$0 \leq b < 1,$ $c > 0$ (only depending on $\mathcal C$) such that for every $m \geqslant 1$, every $m$-edge graph of $\mathcal C$ has a~balanced separator of size at~most $c \cdot m^b$.

Our algorithm also applies to \wmis (\swmis), where vertices have rational weights and the objective is to maximize the total weight of an independent set.
We show the following theorem.

\begin{restatable}{theorem}{mainresult}\label{thm:main}
  Let $\mathcal C$ have strongly edge-sublinear separators. 
  Then \wmis admits a~QPTAS on~$\mathcal C$.
\end{restatable}

For $G\in\mathcal C$ and rational $0<\varepsilon<1$, the algorithm deterministically computes an independent set~$I$ of~$G$ with weight at~least $(1-\varepsilon)\alpha_w(G)$ in time \mbox{$\exp\bigl((\varepsilon^{-1}\log |V(G)|)^{\Oh_{\mathcal C}(1)}\bigr)+\ell^{\Oh(1)}$}, where~$\alpha_w(G)$ is the maximum weight of an independent set of~$G$ and~$\ell$ is the total binary encoding length of the vertex weights and~$\varepsilon$.
No separator or algorithm for finding separators in~$\mathcal C$ is required as part of the input.

The QPTAS is obtained through a~different route than~\cite[Conjecture~1.4.1]{Gartland23} (which anyway cannot hold in general $H$-induced-minor-free classes, let alone in classes with strongly edge-sublinear separators), so the structural induced grid minor conjecture remains open.
By a~separator theorem of Korhonen and Lokshtanov~\cite{Korhonen24}, any class excluding a~fixed graph~$H$ (not necessarily planar) as an induced minor has strongly edge-sublinear separators with $b=\frac{1}{2}+\varepsilon$ for any $\varepsilon > 0$.
Therefore, \cref{thm:main} implies the following corollary.

\begin{corollary}\label{cor:induced-minor}
  For every fixed graph~$H$, \wmis admits a~QPTAS on $H$-induced-minor-free graphs.
\end{corollary}

This unifies the QPTAS guarantee for all the classes listed in the second paragraph, including those admitting faster approximation schemes or exact algorithms, as each excludes a~fixed graph as an induced minor.
We note that having strongly edge-sublinear separators is more general a~property than excluding an induced minor.
For instance, all induced subgraphs of subdivisions of $n$-vertex graphs with at~least $n^{0.1}$ subdivision vertices replacing each edge of the original graph is a~separating example. 

The weighted Nemhauser--Trotter reduction converts a~$(1-\varepsilon)$-approximation for \swmis on a~graph class into a~$(1+\varepsilon)$-approximation for \textsc{Minimum Weight Vertex Cover} on the same class~\cite[Theorem~3.1]{Har-Peled23}.
Thus~\cref{thm:main} also gives a~deterministic QPTAS for this problem on every class with strongly edge-sublinear separators.

The QPTAS guarantee does not extend to all graph optimization problems.
\textsc{Minimum Dominating Set} admits no QPTAS even on $P_5$-free graphs unless $\mathrm{NP}\subseteq\mathrm{QP}$.
This follows from the inapproximability of \textsc{Set Cover}~\cite{Dinur14} via the standard approximation-preserving reduction to split graphs, which are $P_5$-free.
Under the same assumption, \textsc{Maximum Clique} admits no QPTAS on $2K_2$(-induced-minor)-free graphs, by complementation from the APX-hardness of \smis on $C_4$-free graphs~\cite[Theorem~3.6(i)]{Chlebik07Boxes}.

Bonnet and Hickingbotham~\cite{Bonnet25} posed the stronger conjecture that \smis admits a~PTAS when excluding any graph as an induced minor.
This remains open.

\begin{conjecture}[\cite{Bonnet25}]\label{conj:ptas}
  For every graph $H$, \mis admits a~PTAS in $H$-induced-minor-free graphs. 
\end{conjecture}

A~natural strengthening is to ask~\cref{conj:ptas} for classes with strongly edge-sublinear separators, and a~natural weakening is to ask the question for planar graphs~$H$ only.
\Cref{conj:ptas} is even open in the particular case of represented intersection graphs of axis-parallel rectangles.

\subparagraph*{Related work.}
On classes with strongly sublinear separators (measured in vertices), \smis admits a~PTAS~\cite{Har-Peled17} and \swmis a~QPTAS~\cite{Dvorak18,Dvorak21}.
\Cref{thm:main} generalizes the latter to separators measured in edges.
A~PTAS for \swmis is also known on \emph{efficiently fractionally treewidth-fragile} classes~\cite{Dvorak16,Dvorak21}, i.e., where for every fixed $\varepsilon>0$, one can compute in polynomial time a~distribution on vertex sets~$X$ such that $\Pr[v\in X]\leq\varepsilon$ for every vertex~$v$ and each graph~$G-X$ has treewidth bounded in terms of~$\varepsilon$ and the class.

String graphs, the intersection graphs of curves in the plane, have strongly edge-sublinear separators~\cite{Matousek14}.
More generally, Lee~\cite{Lee17} proves such separator bounds for \emph{region intersection graphs} over $K_t$-minor-free hosts, for every fixed~$t$.
These are intersection graphs of connected subgraphs of the host graphs.
\Cref{thm:main} therefore gives a~deterministic QPTAS for \swmis on these classes using only their adjacency relation.
It recovers the schemes for represented intersection graphs of axis-parallel rectangles~\cite{Adamaszek13} and polygons~\cite{Adamaszek19}, and for packing connected subgraphs of a~planar graph~\cite{Pilipczuk18}, without requiring the corresponding representations.

Earlier QPTASes for \swmis under induced-minor exclusion cover $C_t$ for $t\geq4$~\cite[Theorem~1.4]{Chudnovsky24}, $W_\ell$ for $\ell\geq3$~\cite[Corollary~1.3]{Chudnovsky25}, $sC_t$ for $s\geq1$ and $t\geq3$~\cite{Bonnet26d}, and $k\Theta_t$ for $k,t\geq1$~\cite[Theorem~1.11]{Czyzewska26}, with all parameters fixed.
Here~$W_\ell$ is a~cycle of length~$\ell$ with an additional vertex adjacent to every cycle vertex, and~$\Theta_t$ consists of three internally vertex-disjoint paths with common endpoints, each having~$t$ internal vertices.
The notation~$sC_t$ and~$k\Theta_t$ denotes disjoint unions of~$s$ cycles and~$k$ copies of~$\Theta_t$, respectively.
All these results are special cases of~\cref{cor:induced-minor}.
Nevertheless, some of these papers~\cite{Chudnovsky25,Czyzewska26} obtain their QPTAS by showing the still-open structural induced grid minor conjecture (for some planar graphs rather than grids).

\subparagraph*{Proof outline.}
We first deal with unweighted \smis.
The algorithm finds a~balanced separator of small weight with respect to some vertex costs.
We delete the balanced separator, recursively find independent sets in the remaining components, and take their union.
The separator choice and the vertex costs are informed by a~linear programming relaxation. 
The algorithm either finds a~suitable approximation or yields a~valid inequality that is added to the linear program.
This inequality is given to the ellipsoid algorithm.
This technique is called round-or-cut.
We now detail every step.
The reader may want to go through~\cref{alg:solve,alg:round-or-cut} before or after reading the more detailed outline, which follows.

\Cref{sec:separators} presents an algorithm that finds a~balanced separator of~$G$ of small total cost when the vertex costs $x \in [0,1]^{V(G)}$ satisfy some conditions.
(For $X\subseteq V(G)$, we set $x(X):=\sum_{v\in X}x_v$.)
More precisely, for any small positive $\delta$, there is a~threshold~$p=(\delta^{-1}\log n)^{\Oh(1)}$, such that
\begin{equation}\label{eq:outline-separator-conditions}
  x(V(G))=k\geq p
  \qquad\text{and}\qquad
  \sum_{uv\in E(G)}x_ux_v\leq k^2/p
\end{equation}
guarantee a~balanced separator~$S$ with $x(S) \leq \delta k$.
The proof uses a~vertex-capacitated version of the Leighton--Rao flow--cut theorem~\cite{Leighton99}.
If the separator is not found, the resulting flow lets us represent the edges of a~subcubic graph of large treewidth by sampling paths; this follows Korhonen and Lokshtanov~\cite{Korhonen24}.
The assumptions on~$x$ control the lengths, intersections, and adjacencies of these paths.
The short paths induce a~graph violating the assumption that the class has strongly edge-sublinear treewidth.

For a~target $k>p$ for the size of the independent set, \cref{sec:relaxation} introduces a~type of Sherali--Adams relaxation~\cite{Sherali90} with variables~$y_A$ indexed by polylogarithmically small vertex sets; see~\eqref{eq:relaxation}.
For every independent set~$I$ of size~$k$, the assignment $y_A=\mathbf 1[A \subseteq I]$ satisfies its constraints.
For each~$k$, we apply the ellipsoid algorithm to the convex hull of the encodings of independent sets of size~$k$, using the original relaxation constraints and further inequalities obtained from recursive calls.

If the query~$y$ satisfies the constraints of the relaxation, an entropy-decrease argument shows that conditioning on a~set~$A$ with $|A|=\Oh(p\log n)$ and $y_A>0$ gives costs $x_v:=y_{A\cup\{v\}}/y_A$ satisfying~\eqref{eq:outline-separator-conditions}.
(Conditioning on~$A$ fixes $x_v=1$ for every $v \in A$.)
Adding a~vertex to~$A$ requires recomputing all these ratios, so the costs outside~$A$ change, too.

The entropy of the distribution~$(x_v/k)_{v\in V(G)}$ lies between~$0$ and~$\log n$.
Whenever $\sum_{uv\in E(G)}x_ux_v>k^2/p$, conditioning on a~suitable vertex decreases the entropy by more than~$2/p$, which upper-bounds the number of vertices added to~$A$.
We find~$A$ deterministically by enumerating all sets up to this size, which is polylogarithmic with our choice of parameters.
Targets $k \leq p$ are handled by exhaustive search.

In conclusion, \cref{sec:relaxation} gives us vertex costs $x_v$ satisfying \eqref{eq:outline-separator-conditions}, and with~\cref{sec:separators}, we obtain a~balanced separator $S$ such that $x(S) \leq \delta k$ from a~positive query.
In~\cref{sec:algorithm}, we use this separator~$S$ to recursively compute an independent set~$I_C$ and an upper bound $u_C \geq \alpha(G[C])$ for each component~$C$ of~$G-S$.
Let~$\rho$ be a~common approximation factor guaranteed by the child calls, so that $|I_C| \geq \rho u_C$ for every~$C$.
Note that $I := \bigcup_C I_C$ is an independent set because distinct components have no edges between them.
If $x(C)\leq u_C$ for every component $C$, then 
\[
  |I|=\sum_C|I_C|
  \geq\rho\sum_Cu_C
  \geq\rho\sum_Cx(C)
  =\rho(k-x(S))
  \geq\rho(1-\delta)k.
\]
Each component has at most two thirds as many vertices as its parent, so there are at most~$\lceil\log_{3/2}n\rceil$ separator steps along any branch of the recursion tree.
Starting from exact solutions at the leaves, these factors multiply to give an independent set of size at least~$(1-\delta)^{\lceil\log_{3/2}n\rceil}k$.

Otherwise, a~component with $x(C)>u_C$ yields the inequality $\sum_{v\in C} z_{A\cup\{v\}} \leq u_C z_A$ in variables~$z$ with the same coordinates as~$y$.
This inequality holds for every vector $z$ encoding an independent set of size~$k$ but is violated by the current query~$y$.

We pass this inequality to the ellipsoid algorithm, which generates another query or certifies infeasibility.
This is the round-or-cut approach of Carr et~al.~\cite{CarrFLP00}.
The procedure eventually constructs a~sufficiently large independent set or proves $\alpha(G)<k$.
Trying targets in decreasing order establishes the upper bounds required by the recursion: all targets above the first successful one have been ruled out.
Choosing $\delta:=\varepsilon/\lceil\log_{3/2}n\rceil$ bounds the total approximation loss by~$\varepsilon$, since Bernoulli's inequality gives
\[
  (1-\delta)^{\lceil\log_{3/2}n\rceil}
  \geq 1-\delta\lceil\log_{3/2}n\rceil
  =1-\varepsilon.
\]
\Cref{subsec:time} shows that the total running time is quasipolynomial.

\Cref{sec:weights} reduces \swmis to \smis by scaling and rounding the weights to integers, then replacing each vertex by that many pairwise nonadjacent copies.
We prove that this operation preserves strongly edge-sublinear separators.
Applying the unweighted QPTAS to the resulting graph and returning the original vertices with a~selected copy gives the required approximation for \swmis.

\section{Preliminaries}\label{sec:prelim}

We write~$\log$ for the base-2 logarithm.
Subscripts in asymptotic notation indicate which parameters the implicit constants may depend on.
For a~statement~$P$, the indicator~$\mathbf1[P]$ equals~1 if~$P$ holds and~0 otherwise.

All graphs are finite, simple, and undirected unless explicitly stated otherwise.
For a~graph~$G$, we write $V(G)$ and $E(G)$ for its vertex and edge sets.
We denote by~$\alpha(G)$ the \emph{independence number} of $G$, i.e., the size of a~largest independent set of~$G$.
For $X\subseteq V(G)$, we denote by $G[X]$ the subgraph induced by~$X$, and set $G-X:=G[V(G)\setminus X]$.
For $v\in V(G)$, its open neighborhood is $N_G(v):=\{u\in V(G):uv\in E(G)\}$.
An \emph{induced minor} is obtained by vertex deletions and edge contractions.
A~graph is \emph{subcubic} if its maximum degree is at most~3.
A~graph class is \emph{hereditary} if it is closed under taking induced subgraphs.

For a~vector~$x$ indexed by~$V(G)$ with nonnegative entries, we set \[x(X):=\sum_{v\in X}x_v.\]

A~\emph{tree-decomposition} of~$G$ consists of a~tree~$T$ and a~family $\{B_t \subseteq V(G) \mid t \in V(T)\}$ that satisfies the following two conditions:
\begin{itemize}
\item for every $v \in V(G)$, $T[\{t \in V(T) \mid v \in B_t\}]$ is a~nonempty connected subtree of~$T$, and
\item for every $uv \in E(G)$, there is some $t \in V(T)$ such that $\{u,v\} \subseteq B_t$.
\end{itemize}
A~set $B_t$ is called a~\emph{bag}.
The width of the decomposition is $\max_t |B_t|-1$, and the \emph{treewidth}~$\tw(G)$ is the minimum width of a~tree-decomposition of~$G$.

In this paper, for convenience, connected components (or components) are vertex sets (not graphs).
A~\emph{balanced separator} of an $n$-vertex graph~$G$ is a~set $S\subseteq V(G)$ such that every connected component of $G-S$ has at most $2n/3$ vertices.
By the \emph{components of a~separator}~$S$, we mean the connected components of~$G-S$.
The \emph{separation number}~$\sn(G)$ is the least integer~$s$ such that every (induced) subgraph of~$G$ has a~balanced separator of size at most~$s$.
Treewidth and separation number are linearly tied.
The upper bound on the separation number is due to Robertson and Seymour~\cite{Robertson86a}, and the converse is due to Dvo\v{r}\'ak and Norin~\cite{Dvorak19}.

\begin{theorem}\label{thm:treewidth}
  For every graph~$G$ with at least one edge, $\tw(G)=\Theta(\sn(G))$.
\end{theorem}

Let $\mathcal C$ be a~class satisfying the assumption of~\cref{thm:main} for some exponent $0\leq b<1$.
Choose an integer $a \geq 2$ such that $b \leq 1-1/a$.
The preceding theorem therefore gives
\begin{equation}\label{eq:treewidth}
  \tw(G)=\Oh_{\mathcal C}(|E(G)|^{1-1/a})
  \qquad\text{for every }G\in\mathcal C.
\end{equation}
For edgeless~$G$, the bound is immediate.
We fix~$\mathcal C$ and~$a$ for the rest of the proof.

\section{Light balanced separators}\label{sec:separators}

For a~graph $G$, a~nonnegative vector $x$ indexed by~$V(G)$, and $0<\delta<1$, a~balanced separator~$S$ of~$G$ is \emph{$(\delta,x)$-light} if $x(S)\leq\delta x(V(G))$.
Note that the balance is measured by the number of vertices in each component, and lightness by the cost~$x$.
We set
\[
  e_G(x):=\sum_{uv\in E(G)}x_ux_v.
\]
The goal of this section is to find a~$(\delta,x)$-light balanced separator when $x(V(G))$ is sufficiently large and $e_G(x)$ is relatively small.

\begin{lemma}\label{lem:cheap-separator}
  Let~$\mathcal C$ be a~graph class with strongly edge-sublinear separators, and let~$a \geq 2$ be an integer satisfying~\eqref{eq:treewidth}.
  For every integer $n \geq 2$ and rational $0<\delta<1$, there is an integer
  \begin{equation}\label{eq:separator-parameters}
    p=\Oh_{\mathcal C}\bigl((\delta^{-1}\log n)^{4a-2}\bigr)
  \end{equation}
  such that for any $n$-vertex graph $G\in\mathcal C$ and rational vector $x\in[0,1]^{V(G)}$ with
  \[k:=x(V(G))\geq p~~\text{and}~~e_G(x)\leq k^2/p,\]
  one can find a~$(\delta,x)$-light balanced separator of~$G$.
  The algorithm is deterministic and takes polynomial time in the input encoding length.
\end{lemma}

We use a~vertex-capacitated version of the Leighton--Rao multicommodity flow--cut theorem~\cite{Leighton99}, with capacities proportional to the costs~$x_v$.
Repeated applications yield either the required balanced separator or an all-pairs flow in an induced subgraph.
In the latter case, we sample paths from this flow to represent the edges of a~subcubic graph of large treewidth.
The capacity constraints and the assumptions on~$x$ control the sizes of these paths and bound the probability that paths representing nonincident edges intersect or are joined by an edge.
With positive probability, discarding long paths leaves a~collection whose vertices induce a~graph with large treewidth but too few edges to satisfy~\eqref{eq:treewidth}.
Heredity ensures that this graph belongs to~$\mathcal C$, giving a~contradiction.

The construction of an \emph{induced almost-embedding} (see~\cref{subsec:paths} for definitions) from a~concurrent flow follows Korhonen and Lokshtanov~\cite[Section~4.2]{Korhonen24}.
Earlier separator arguments using multicommodity flows include those of Matou\v{s}ek~\cite{Matousek14} and Biswal, Lee, and Rao~\cite{BiswalLeeRao2010}.
In particular, the latter already separates long and short flow paths in the treatment of shallow excluded minors: many long paths force large congestion, while restricting to short paths permits a~shallow-minor argument.
Here we combine the vertex-dependent bounds in~\eqref{eq:path-probability}, the assumptions on~$x$, and a~subcubic graph of linear treewidth.
These ingredients let us discard edges assigned to long paths while preserving linear treewidth.

\subsection{Flows and separations}\label{subsec:flows}

A~\emph{separation} of~$G$ is a~pair $(A,B)$ of nonempty vertex sets with $A\cup B=V(G)$ and with no edge between $A\setminus B$ and $B\setminus A$.
Its separator is $A\cap B$.

We use nonnegative vertex costs~$x_v$, so the cost of a~separator~$S$ is~$x(S) = \sum_{v \in S} x_v$.
To find a~separator of small cost, we apply a~multicommodity flow--cut theorem with vertex capacities proportional to these costs.

\subparagraph*{Flows and capacity constraints.}

For $u,v\in V(G)$, let~$\mathcal P_{u,v}(G)$ be the set of simple paths from~$u$ to~$v$, regarded as ordered vertex sequences.
In particular, $\mathcal P_{v,v}(G):=\{(v)\}$.
A~\emph{unit all-pairs multicommodity flow} assigns a~nonnegative number~$f(P)$ to each such path, satisfying
\[
  \sum_{P\in\mathcal P_{u,v}(G)}f(P)=1
  \qquad\text{for every }(u,v)\in V(G)^2.
\]
Its \emph{load} at~$w\in V(G)$ is
\[
  \operatorname{load}_f(w):=
    \sum_{u,v\in V(G)}\ \sum_{\substack{P\in\mathcal P_{u,v}(G)\\w\in V(P)}}f(P).
\]

For a~given~$\gamma>0$, the vertex-capacity constraints used below are
\[
  \operatorname{load}_f(v)\leq\gamma x_v
  \qquad(v\in V(G)).
\]

\subparagraph*{The flow--cut theorem.}
The following is a~consequence of the directed Leighton--Rao flow--cut theorem~\cite[Theorem~12]{Leighton99} and the vertex-splitting reduction~\cite[Section~3.5]{Leighton99}.
We adapt the proof to our convention that vertex loads include endpoints and that every ordered pair, including~$(v,v)$, has unit demand.

\begin{lemma}\label{lem:vertex-flow}
  Given an $n$-vertex graph~$G$ with $n\geq2$, nonnegative rational vertex costs~$x_v$, and a~rational $\gamma > 0$, there is a~deterministic algorithm, polynomial in the encoding of $(G, x, \gamma)$, that finds
  \begin{enumerate}
    \item a~unit all-pairs multicommodity flow~$f$ satisfying $\operatorname{load}_f(v)\leq\gamma x_v$ for every $v\in V(G)$, or
    \item a~separation $(A,B)$ such that
      \begin{equation}\label{eq:vertex-cut}
        x(A\cap B)=\Oh\left(\frac{|A||B|}{\gamma}\log n\right).
      \end{equation}
  \end{enumerate}
\end{lemma}
\begin{proof}
  Split each vertex~$v$ into~$v^-,v^+$ and give the arc $v^-\to v^+$ capacity~$x_v$.
  Each edge~$uv$ gives arcs $u^+\to v^-$ and $v^+\to u^-$.
  These arcs and all reverse arcs $v^+\to v^-$ receive capacity~$c$, a~sufficiently large integer multiple of~$\gamma^{-1}n^2\lceil\log n\rceil$.
  This finite value acts as infinite capacity: it exceeds both the total demand and the cut-capacity bound below.
  Run the directed flow--cut algorithm of Leighton and Rao~\cite[proof of Theorem~12]{Leighton99} on this network~$\vec G$, with demand~$1/\gamma$ for every ordered pair of distinct split vertices.

  If all demands are routed, retain the $u^-$-to-$v^+$ flows for all $(u,v)\in V(G)^2$ and decompose them into simple directed paths; the split endpoints are distinct even when $u=v$.
  Such a~path cannot use a~reverse arc $w^+\to w^-$: it reaches~$w^+$ through~$w^-$, so using the reverse arc would repeat a~vertex.
  Every retained path consists of consecutive pairs $(w^-,w^+)$, each traversing the capacity arc \mbox{$w^-\to w^+$}, joined by arcs corresponding to edges of~$G$.
  Replacing each pair $(w^-,w^+)$ by~$w$ gives a~simple path in~$G$ with the same original endpoints; when $u=v$, this is the singleton path~$(u)$.
  Multiply the flow amount on each path by~$\gamma$.
  This yields our unit all-pairs flow with $\operatorname{load}_f(w) \leq \gamma x_w$.

  Otherwise, we obtain a~nontrivial cut~$(U,\overline U)$ of capacity $\Oh(\gamma^{-1} |U||\overline U|\log n)<c$.
  Put $A:=\{v:v^-\in U\}$ and $B:=\{v:v^+\in\overline U\}$.
  No arc of capacity~$c$ leaves~$U$, so the reverse arcs and edge arcs ensure that~$(A,B)$ is a~separation.
  Its cost~$x(A\cap B)$ equals the cut capacity; the bounds $|U| \leq 2|A|$ and $|\overline U| \leq 2|B|$ then give~\eqref{eq:vertex-cut}.
  The algorithm in~\cite[proof of Theorem~12]{Leighton99} solves the multicommodity flow linear program and, if the demands cannot be routed, constructs a~cut deterministically from an optimal dual solution.
  The network has polynomial size and rational capacities of polynomial encoding length, so these steps and projecting the flow onto~$G$ take polynomial time in the input encoding length.
\end{proof}

\subparagraph*{From separations to a light balanced separator.}
Given an $n$-vertex graph~$G$ with $n\geq2$ and a~nonnegative and nonzero (every entry is nonnegative and not all entries are 0) rational vector~$x$, we will always set $k:=x(V(G))>0$.
For any rational $0<\delta<1$, integer $\ell$ is picked to be at least a~sufficiently large universal constant times~$\delta^{-1}\log n$, and we set
\begin{equation}\label{eq:flow-parameters}
  \gamma:=\frac{\ell n^2}{3k}.
\end{equation}

We repeatedly apply~\cref{lem:vertex-flow} to show the following lemma (as in~\cite[Lemma~4]{Bonnet26b}).

\begin{samepage}
\begin{lemma}\label{lem:balanced-flow}
  With~$G,x,\delta,\ell,\gamma$ chosen as above, there is a~deterministic algorithm, polynomial in their encoding length, that finds
  \begin{enumerate}
    \item a~$(\delta,x)$-light balanced separator of~$G$, or\label{item:balanced-flow-separator}
    \item an induced subgraph~$G'$\label{item:balanced-flow-paths} with more than $2n/3$ vertices and a~probability distribution on the ordered simple paths of~$G'$.
      If~$Q$ is drawn from this distribution, then
      \[
        \Pr[Q\in\mathcal P_{u,v}(G')]=\frac{1}{|V(G')|^2}
        \qquad\text{for every }(u,v)\in V(G')^2.
      \]
      Moreover,
      \begin{equation}\label{eq:path-probability}
        \Pr[v\in V(Q)]\leq \frac{\ell x_v}{k}
        \qquad\text{for every }v\in V(G').
      \end{equation}
  \end{enumerate}
\end{lemma}
\end{samepage}
\begin{proof}
  Start with $G':=G$ and $S:=\varnothing$.
  While $|V(G')|>2n/3$, apply~\cref{lem:vertex-flow} to~$G'$, using the restriction of~$x$ and the fixed parameter~$\gamma$.
  Stop if a~flow is returned.
  Otherwise, swap the sides of the returned separation $(A_i,B_i)$ if necessary so that $|A_i|\geq|B_i|$.
  Add $A_i\cap B_i$ to~$S$ and replace~$G'$ by $G'[A_i\setminus B_i]$.
  Each such step removes the nonempty set~$B_i$, so these sets are pairwise disjoint and there are at most~$n$ calls to~\cref{lem:vertex-flow}.

  By~\eqref{eq:vertex-cut}, the total separator cost is
  \begin{equation}\label{eq:separator-cost}
    x(S)=\Oh\left(\frac{\log n}{\gamma}\sum_i|A_i||B_i|\right)
      =\Oh\left(\frac{n^2\log n}{\gamma}\right)
      =\Oh\left(\frac{3 k \log n}{\ell}\right) \leq \delta k.
  \end{equation}
  Here $|A_i|\leq n$ and $\sum_i|B_i|\leq n$ give the second bound, and the choice of~$\ell$ gives the last inequality.
  If every call returns a~separation, the discarded sets $B_i\setminus A_i$ and the final set~$V(G')$ partition $V(G)\setminus S$, with no edges between distinct parts.
  Since $|A_i\setminus B_i|\geq|B_i\setminus A_i|$, each discarded set has at most~$n/2$ vertices, and the final set has at most~$2n/3$ vertices.
  Thus~$S$ is a~$(\delta,x)$-light balanced separator.

  If a~unit all-pairs flow~$f$ is returned on~$G'$, define~$Q$ by the following sampling rule.
  Choose~$u,v$ independently and uniformly from~$V(G')$.
  Conditional on this choice, select a~path $P\in\mathcal P_{u,v}(G')$ with probability~$f(P)$, and set $Q:=P$.
  The unit-flow equations ensure that these conditional probabilities sum to~1.
  Consequently, every ordered path~$P$ in~$G'$ satisfies
  \[
    \Pr[Q=P]=\frac{f(P)}{|V(G')|^2}.
  \]
  Thus, for every vertex~$w\in V(G')$,
  \[
    \Pr[w\in V(Q)]
      =\frac{\operatorname{load}_f(w)}{|V(G')|^2}
      \leq\frac{\gamma x_w}{|V(G')|^2}
      \leq\frac{\ell x_w}{k},
  \]
  where the last inequality uses $|V(G')|>2n/3$ and $(3/2)^2 < 3$.
  This gives item~\ref{item:balanced-flow-paths}.
  Its distribution is specified deterministically by the rational probabilities~$f(P)/|V(G')|^2$ on the polynomially many paths carrying a~positive flow.
  The running time follows from~\cref{lem:vertex-flow} and the bound on the number of iterations.
\end{proof}

\subsection{Treewidth and sampled paths}\label{subsec:paths}

We rule out item~\ref{item:balanced-flow-paths} of~\cref{lem:balanced-flow} by sampling paths for the edges of a~$t$-vertex subcubic graph of treewidth~$\Omega(t)$, following Korhonen and Lokshtanov~\cite[Section~4.2]{Korhonen24}.
Using~\eqref{eq:path-probability}, we bound the probability that paths assigned to nonincident edges intersect or are joined by an edge.
Summing~\eqref{eq:path-probability} over vertices also bounds the expected number of vertices on each sampled path.
After deleting the edges assigned to long paths, the remaining path vertices induce a~graph~$F$ with treewidth~$\Omega(t)$ and~$\Oh(t\ell^2)$ edges.
For a~suitable choice of~$t$, these bounds contradict~\eqref{eq:treewidth}.

An~\emph{almost-embedding} of~$H$ in~$G$ is a~pair~$(\phi,\pi)$, where $\phi:V(H)\to V(G)$ and~$\pi$ assigns to each edge~$uv$ of~$H$ a~path from~$\phi(u)$ to~$\phi(v)$ in~$G$, such that paths assigned to \emph{nonincident} edges (i.e., with no common endpoint) are vertex-disjoint.
The almost-embedding is \emph{induced} if there is also no edge of~$G$ between paths assigned to nonincident edges.
The map $\phi$ need not be injective, and singleton paths are allowed images of~$\pi$.

The following observation transfers a~tree-decomposition through an almost-embedding.

\begin{lemma}\label{lem:path-treewidth}
  Let~$H$ have maximum degree at most~$d$ and no isolated vertices.
  If~$H$ has an almost-embedding in~$G$, then every tree-decomposition of~$G$ with bags of size at most~$w$ gives a~tree-decomposition of~$H$ with bags of size at most~$2dw$.
\end{lemma}
\begin{proof}
  Let~$(\phi,\pi)$ be an almost-embedding of~$H$ in~$G$.
  For each $v \in V(G)$, let~$L_v$ consist of the endpoints of the edges $e \in E(H)$ such that $v \in V(\pi(e))$.
  We first claim that $|L_v| \leq 2d$.
  Indeed, if some path~$\pi(xy)$ contains~$v$, every edge whose path contains~$v$ shares an endpoint with~$xy$, since paths assigned to nonincident edges are vertex-disjoint.
  All endpoints of these edges belong to $N_H(x) \cup N_H(y)$.
  This union has size at most~$2d$.

  Take a~tree-decomposition $(T,(B_t)_{t\in V(T)})$ of~$G$ and replace each bag $B_t$ by
  \[
    B'_t:=\bigcup_{v \in B_t} L_v.
  \]
  Every vertex of~$H$ belongs to a~new bag, since it has an incident edge.
  Both endpoints of an edge~$xy \in E(H)$ belong to every new bag whose original bag intersects~$\pi(xy)$.
  For a~vertex~$x$ of~$H$, the union of the paths representing edges incident with~$x$ is connected, since all of them contain~$\phi(x)$.
  In a~tree-decomposition, the bags meeting a~connected vertex set form a~connected subtree: the subtrees for adjacent vertices intersect in a~bag containing their edge.
  Thus the bags~$B'_t$ containing~$x$ form a~connected subtree.
  They give a~tree-decomposition of~$H$, and $|B'_t|\leq2d|B_t|$ proves the bound.
\end{proof}

To apply this observation, we use the following existence result.

\begin{lemma}[{\cite{Hoory06}, \cite[Proposition~1]{Grohe09}}]\label{lem:expanders}
  For every sufficiently large integer~$t$, there is a~$t$-vertex subcubic graph of treewidth~$\Omega(t)$.
\end{lemma}

Subcubic expanders exist~\cite{Hoory06}, and their expansion implies linear treewidth~\cite[Proposition~1]{Grohe09}.
We can now show the main novel technical lemma of this section. 

\begin{lemma}\label{lem:no-flow}
  For every $\ell\geq1$, there is an integer $p=\Oh_{\mathcal C}(\ell^{4a-2})$ with the following property.
  Let $G\in\mathcal C$ and $x\in[0,1]^{V(G)}$, and put $k:=x(V(G))$.
  If $k\geq p$ and $e_G(x)\leq k^2/p$, then~$G$ does not satisfy item~\ref{item:balanced-flow-paths} of~\cref{lem:balanced-flow}.
\end{lemma}
\begin{proof}
  Suppose that item~\ref{item:balanced-flow-paths} of~\cref{lem:balanced-flow} holds, with induced subgraph~$G'$ and random path~$Q$.
  Choose an integer~$t$ that is a~sufficiently large constant multiple of~$\ell^{2a-2}$, and let~$H$ be the graph yielded by~\cref{lem:expanders}.
  We will specify how large this multiple needs to be at the end of the proof.
  We then choose an integer $p=\Oh(t^2\ell^2)$ with a~sufficiently large implicit constant, which gives the required bound on~$p$.

  Fix an orientation of each edge of~$H$.
  Choose~$\phi(u)$, for $u\in V(H)$, independently and uniformly from~$V(G')$.
  Conditional on~$\phi$, choose the paths~$\pi(uv)$ independently over the oriented edges~$uv$, with
  \[
    \Pr[\pi(uv)=P\mid\phi] 
      =|V(G')|^2\Pr[Q=P]
    \qquad\text{for }P\in\mathcal P_{\phi(u),\phi(v)}(G').
  \]
  These probabilities sum to~1 because the ordered endpoints of~$Q$ are independent and uniform in~$V(G')$.
  Each~$\pi(uv)$ therefore has the same distribution as~$Q$.
  For nonincident edges, the endpoint assignments are independent, and the paths are sampled independently conditional on those assignments.
  Hence the two paths are independent.

  Let $e,f$ be nonincident edges of~$H$.
  By the union bound and independence,
  \begin{align}
    \Pr[V(\pi(e)) \cap V(\pi(f))\neq\varnothing]
      &\leq\sum_{v\in V(G')}\Pr[v\in V(\pi(e))]\cdot\Pr[v\in V(\pi(f))]\notag\\
      &\leq\frac{\ell^2}{k^2}\sum_{v\in V(G')}x_v^2\notag\\
      &\leq\frac{\ell^2}{k^2}\sum_{v\in V(G)}x_v
       =\frac{\ell^2}{k},\label{eq:intersection}
  \end{align}
  where we use~\eqref{eq:path-probability}, $x_v^2\leq x_v$, and $x(V(G))=k$.
  Moreover,
  \begin{align}
    &\Pr[\text{an edge joins }V(\pi(e))\text{ and }V(\pi(f))]\notag\\
    &\quad\leq\sum_{uv\in E(G')}\Bigl(
      \Pr[u\in V(\pi(e))]\cdot\Pr[v\in V(\pi(f))]
      +\Pr[v\in V(\pi(e))]\cdot\Pr[u\in V(\pi(f))]\Bigr)\notag\displaybreak[1]\\
    &\quad\leq\frac{2\ell^2}{k^2}\sum_{uv\in E(G')}x_ux_v\notag\\
    &\quad\leq\frac{2\ell^2}{k^2}\sum_{uv\in E(G)}x_ux_v
      \leq  \frac{2\ell^2}{p}.\label{eq:adjacency}
  \end{align}
  The two summands account for the two choices of which path contains which endpoint of an edge.
  We again use the union bound, independence, and~\eqref{eq:path-probability}; the last inequality follows from the lemma assumption that $e_G(x)\leq k^2/p$.
  Since $|E(H)|\leq3t/2$ and $k\geq p$, the probability of any intersection or adjacency between paths representing nonincident edges is
  \begin{equation}\label{eq:bad-pairs}
    \Oh\left(t^2\left(\frac{\ell^2}{k} + \frac{2\ell^2}{p}\right)\right) = \Oh(t^2\ell^2/p)\leq1/4,
  \end{equation}
  by our choice of~$p$.

  Summing~\eqref{eq:path-probability} over vertices gives
  \[
    \E[|V(\pi(e))|]\leq\frac{\ell}{k}\sum_{v\in V(G')}x_v\leq\ell.
  \]
  By linearity of expectation and Markov's inequality,
  \[
    \Pr\left[\sum_{e\in E(H)}|V(\pi(e))|\leq4\ell|E(H)|\right]\geq\frac34.
  \]
  Fix a~universal constant~$c$ (not depending on~$t$) such that $8|E(H)|\leq c \cdot \tw(H)$.
  
  Whenever $\sum_{e\in E(H)}|V(\pi(e))|\leq4\ell|E(H)|$, dividing the right-hand side by~$c\ell$ shows that at most $4|E(H)|/c\leq\tw(H)/2$ paths have more than~$c\ell$ vertices.
  Together with~\eqref{eq:bad-pairs}, this gives positive probability that~$(\phi,\pi)$ is an induced almost-embedding of~$H$ in~$G'$ and that at most~$\tw(H)/2$ paths exceed this size (indeed, both events have probability at least $3/4$).
  Fix such an outcome.

  Delete every edge~$e$ of~$H$ for which $|V(\pi(e))|>c\ell$, and remove the resulting isolated vertices, obtaining~$H'$.
  Adding one edge to a~graph increases its treewidth by at most~1: add one endpoint of that edge to every bag of a~tree-decomposition.
  Removing isolated vertices preserves positive treewidth.
  Consequently,
  \begin{equation}\label{eq:remaining-treewidth}
    \tw(H')\geq\tw(H)/2=\Omega(t).
  \end{equation}
  Let
  \[
    F:=G\left[\bigcup_{e\in E(H')}V(\pi(e))\right].
  \]
  Restricting~$\phi$ and~$\pi$ to~$H'$ gives an induced almost-embedding of~$H'$ in~$F$.

  Every retained path has at most~$c\ell$ vertices.
  Since the almost-embedding is induced, every edge of~$F$ has both endpoints on one retained path or joins paths assigned to incident edges of~$H'$.
  There are at most~$3t/2$ retained paths and at most~$3t$ pairs of incident edges, since~$H'$ is subcubic.
  Each path or pair of paths contains~$\Oh(\ell)$ vertices and therefore accounts for at most~$\Oh(\ell^2)$ edges in~$F$.
  Hence
  \begin{equation}\label{eq:induced-edges}
    |E(F)|=\Oh(t\ell^2).
  \end{equation}
  On the other hand, \cref{lem:path-treewidth} with~$d=3$ and~\eqref{eq:remaining-treewidth} give $\tw(F)=\Omega(t)$.
  The graph~$F$ is induced in~$G$ and hence belongs to~$\mathcal C$.
  Applying~\eqref{eq:treewidth}, we obtain
  \[
    t=\Oh_{\mathcal C}\bigl((t\ell^2)^{1-1/a}\bigr),
    \qquad\text{and thus}\qquad
    t=\Oh_{\mathcal C}(\ell^{2a-2}).
  \]
  Choosing the constant in the definition of~$t$ larger than the last implicit constant gives a~contradiction.
  The choices of~$t$ and then~$p$ depend only on~$\mathcal C$ and~$\ell$, as required.
\end{proof}

We can now conclude.

\begin{proof}[Proof of~\cref{lem:cheap-separator}]
  Choose $\ell:=\lceil c_0\delta^{-1}\log n\rceil$, where~$c_0$ is a~sufficiently large universal constant for~\cref{lem:balanced-flow}, and set~$\gamma$ by~\eqref{eq:flow-parameters}.
  Let~$p$ be the integer given by~\cref{lem:no-flow}.
  Then $p=\Oh_{\mathcal C}(\ell^{4a-2})$, which gives~\eqref{eq:separator-parameters}.
  Apply~\cref{lem:balanced-flow}.
  Item~\ref{item:balanced-flow-paths} contradicts~\cref{lem:no-flow}, so item~\ref{item:balanced-flow-separator} gives a~$(\delta,x)$-light balanced separator.
  The sampling in~\cref{lem:no-flow} existentially shows that item~\ref{item:balanced-flow-paths} cannot hold, but will not be used in the overall algorithm; from this section, we will only rely on the deterministic algorithm of~\cref{lem:balanced-flow}.
\end{proof}

Increasing~$p$ strengthens the assumptions of~\cref{lem:cheap-separator}, so the conclusion holds for every larger threshold.
By~\eqref{eq:separator-parameters}, we may therefore choose~$p$ nondecreasing in~$n$.

\section{A~linear relaxation and conditioning}\label{sec:relaxation}

To apply~\cref{lem:cheap-separator} to a~target~$k \geq p$, we need vertex costs~$x_v\in[0,1]$ satisfying
\[
  x(V(G))=k
  \qquad\text{and}\qquad
  e_G(x)\leq k^2/p.
\]
We obtain these costs from a~feasible point of a~linear relaxation by conditioning on the inclusion of a~small vertex set.

Fix an $n$-vertex graph~$G$, an integer $1 \leq k \leq n$, and an integer~$r \geq 1$.
Introduce a~variable~$y_A$ for every $A \subseteq V(G)$ with $|A| \leq r+2$.
Consider the following system of linear constraints.
\begin{equation}\label{eq:relaxation}
  \begin{aligned}
    y_\varnothing&=1,\\
    0\leq y_A&\leq1
      &&(|A|\leq r+2),\\
    y_{A\cup\{v\}}&\leq y_A
      &&(|A|\leq r+1,\ v\in V(G)),\\
    y_A&=0
      &&(|A|\leq r+2,\ A\text{ is not an independent set}),\\
    \sum_{v\in V(G)}y_{A\cup\{v\}}&=ky_A
      &&(|A|\leq r+1).
  \end{aligned}
\end{equation}
Let $R_r(G,k)$ be the polytope defined by~\eqref{eq:relaxation}.
Every independent set~$I$ of size~$k$ gives a~point of this polytope by setting
\[
  y_A:=\mathbf1[A\subseteq I].
\]
The first four rows hold directly.
For the last row, observe that
\[
  \sum_{v\in V(G)}\mathbf1[A\cup\{v\}\subseteq I]
  =\mathbf1[A\subseteq I]\sum_{v\in V(G)}\mathbf1[v\in I]
  =k \cdot \mathbf1[A\subseteq I].
\]
The indexed variables replace products of 0--1 indicators, as in the Sherali--Adams construction~\cite{Sherali90}.

The number of variables is
\begin{equation}\label{eq:dimension}
  d:=\sum_{j=0}^{\min\{n,r+2\}}\binom{n}{j}
     =n^{\Oh(r)}.
\end{equation}
There are $n^{\Oh(r)}$ constraints, which can all be checked in that time.

\subsection{Conditioning}\label{subsec:conditioning}

Fix a~rational vector~$y\in R_r(G,k)$.
For $A\subseteq V(G)$ with $|A|\leq r$ and $y_A>0$, define
\[
  x_v:=\frac{y_{A\cup\{v\}}}{y_A}.
\]
The constraints $0\leq y_{A\cup\{v\}}\leq y_A$ in~\eqref{eq:relaxation} give $0\leq x_v\leq1$, and the last row of~\eqref{eq:relaxation} gives
\[
  \sum_{v\in V(G)}x_v
    =\frac{1}{y_A}\sum_{v\in V(G)}y_{A\cup\{v\}}
    =k.
\]
Thus the numbers~$x_v/k$ define a~probability distribution on~$V(G)$.
For a~vertex~$v$ with $x_v>0$, adding~$v$ to~$A$ preserves the positivity of the denominator, since $y_{A\cup\{v\}}>0$.
We call adding~$v$ to~$A$ and recomputing the ratios \emph{conditioning} on~$v$.
Conditioning on a~small set of variables is used in the rounding algorithms of Barak, Raghavendra, and Steurer~\cite{Barak11} and Raghavendra and Tan~\cite{RT12}.

\begin{lemma}\label{lem:conditioning}
  Let~$p\geq1$ be an integer and let~$y\in R_r(G,k)$ be rational.
  If $r\geq2p\log n$, then a~deterministic algorithm finds a~set $A\subseteq V(G)$, with $|A|\leq r$ and $y_A>0$, such that the vector
  \[
    x := \left(\frac{y_{A\cup\{v\}}}{y_A}\right)_{v \in V(G)}
  \]
  satisfies
  \[
    x\in[0,1]^{V(G)},\qquad x(V(G))=k,\qquad e_G(x)\leq k^2/p.
  \]
  Its running time is polynomial in~$r$ and in the encoding lengths of~$y$ and~$p$.
\end{lemma}
\begin{proof}
  We first prove that a~suitable set~$A$ exists.
  Fix~$A$ with $|A|<r$ and $y_A>0$, and write
  \[
    \pi_u:=\frac{x_u}{k},
    \qquad q_{uv}:=\frac{y_{A\cup\{u,v\}}}{k^2y_A}.
  \]
  Applying the last row of~\eqref{eq:relaxation} at~$A\cup\{v\}$ gives
  \begin{equation}\label{eq:marginals}
    \sum_uq_{uv}
      =\frac{1}{k^2y_A}\sum_u y_{A\cup\{u,v\}}
      =\frac{y_{A\cup\{v\}}}{ky_A}
      =\pi_v.
  \end{equation}
  By symmetry, $\sum_vq_{uv}=\pi_u$.
  Therefore~$q$ is a~probability distribution on~$V(G)^2$, with both marginals equal to~$\pi$.
  For each~$v$ with $\pi_v>0$, let~$\pi^v$ be the distribution associated with~$A\cup\{v\}$.
  Its coordinates are
  \begin{equation}\label{eq:conditioned-distribution}
    \pi^v_u:=\frac{y_{A\cup\{u,v\}}}{k y_{A\cup\{v\}}}
         =\frac{q_{uv}}{\pi_v}.
  \end{equation}

  Our proof uses an entropy-decrease argument similar to that of Raghavendra and Tan~\cite[Section~4]{RT12}, applied here to the distribution~$\pi$ on~$V(G)$.
  We show that, whenever $e_G(x)>k^2/p$, conditioning on some vertex decreases entropy by more than~$2/p$.
  Since entropy lies between~0 and~$\log n$, this bounds the number of successive choices needed to obtain a~suitable set~$A$.
  For a~probability distribution~$\pi$ on~$V(G)$, its \emph{entropy} is
  \[
    H(\pi):=-\sum_{u:\pi_u>0}\pi_u\log\pi_u.
  \]
  Every summand is nonnegative, since $0<\pi_u\leq1$, so $H(\pi)\geq0$.
  For the upper bound, Jensen's inequality for the concave function~$\log$, with weights~$\pi_u$ summing to~1, gives
  \[
    \begin{aligned}
      H(\pi)
        &=\sum_{u:\pi_u>0}\pi_u\log\frac{1}{\pi_u}\\
        &\leq\log\left(\sum_{u:\pi_u>0}\pi_u\frac{1}{\pi_u}\right)
          =\log\bigl|\{u:\pi_u>0\}\bigr|\leq\log n.
    \end{aligned}
  \]
  Expanding the entropy definitions, we obtain
  \begin{equation}\label{eq:entropy-drop}
    \begin{aligned}
      H(\pi)-\sum_{v:\pi_v>0}\pi_v H(\pi^v)
        &=-\sum_{u:\pi_u>0}\pi_u\log\pi_u
          +\sum_{u,v:q_{uv}>0}q_{uv}\log\frac{q_{uv}}{\pi_v}\\
        &=\sum_{u,v:q_{uv}>0}q_{uv}
          \left(\log\frac{q_{uv}}{\pi_v}-\log\pi_u\right)\\
        &=\sum_{u,v:q_{uv}>0}q_{uv}\log\frac{q_{uv}}{\pi_u\pi_v}.
    \end{aligned}
  \end{equation}
  The first equality uses~\eqref{eq:conditioned-distribution}, and the second uses the marginal identity $\sum_vq_{uv}=\pi_u$ from~\eqref{eq:marginals} and symmetry.
  The sum on the left averages the entropies after conditioning on a~vertex, with weights~$\pi_v$.

  Since the positive entries of~$q$ sum to~1, Jensen's inequality gives
  \[
    \sum_{u,v:q_{uv}>0}q_{uv}\log\frac{q_{uv}}{\pi_u\pi_v}
      =-\sum_{u,v:q_{uv}>0}q_{uv}\log\frac{\pi_u\pi_v}{q_{uv}}
      \geq-\log\left(\sum_{u,v:q_{uv}>0}\pi_u\pi_v\right).
  \]
  If $uv\in E(G)$, then $q_{uv}=0$ by~\eqref{eq:relaxation}.
  Thus, since $\sum_{u,v}\pi_u\pi_v=(\sum_u\pi_u)(\sum_v\pi_v)=1\cdot1=1$,
  \[
    \begin{aligned}
      \sum_{u,v:q_{uv}>0}\pi_u\pi_v
        &=1-\sum_{u,v:q_{uv}=0}\pi_u\pi_v\\
        &\leq1-\sum_{uv\in E(G)}(\pi_u\pi_v+\pi_v\pi_u)\\
        &=1-\frac{2}{k^2}\sum_{uv\in E(G)}x_ux_v
          =1-\frac{2e_G(x)}{k^2}.
    \end{aligned}
  \]
  The two summands for each edge account for the ordered pairs~$(u,v)$ and~$(v,u)$.
  Combining these inequalities yields
  \begin{equation}\label{eq:entropy-edge-sum}
    H(\pi)
      -\sum_{v:\pi_v>0}\pi_v H(\pi^v)
      \geq-\log\left(1-\frac{2e_G(x)}{k^2}\right)
      \geq\frac{2e_G(x)}{k^2}.
  \end{equation}
  The last inequality uses $-\log(1-t)\geq t$ for $0\leq t<1$.
  The logarithm in~\eqref{eq:entropy-edge-sum} has a~positive argument, since $1-2e_G(x)/k^2\geq\sum_{u,v:q_{uv}>0}\pi_u\pi_v>0$.
  The sum is positive because~$q$ has total mass~1 and $q_{uv}>0$ implies $\pi_u\pi_v>0$.
  Since the weights~$\pi_v$ sum to~1, \eqref{eq:entropy-edge-sum} implies that if $e_G(x)>k^2/p$, some vertex~$v$ with $\pi_v>0$ satisfies
  \[
    H(\pi^v)<H(\pi)-2/p.
  \]

  Start with $A:=\varnothing$ and, while the required inequality fails, add such a~vertex~$v$.
  Every denominator remains positive, and each choice increases~$|A|$ by~1.
  Since entropy starts at most~$\log n$ and remains nonnegative, fewer than~$r$ such choices are possible under the assumption $r\geq2p\log n$.
  Hence a~set~$A$ satisfying the conclusion exists.

  To find one deterministically, examine all sets $A\subseteq V(G)$ of size at most~$r$ with $y_A>0$, and test
  \[
    p\sum_{uv\in E(G)}y_{A\cup\{u\}}y_{A\cup\{v\}}
      \leq k^2y_A^2.
  \]
  This is equivalent to $e_G(x)\leq k^2/p$.
  Upon finding a~set~$A$ that passes the test, set $x_v:=y_{A\cup\{v\}}/y_A$ for every $v\in V(G)$ and return~$(A,x)$.
  All required coordinates occur in~$y$, and there are at most~$d$ candidate sets, where~$d$ is given by~\eqref{eq:dimension}.
  Each test uses $\Oh(n^2)$ rational operations on numbers of polynomial encoding length.
  This proves the running-time bound.
\end{proof}

\section{The approximation algorithm}\label{sec:algorithm}

We combine \cref{lem:cheap-separator,lem:conditioning} in a~recursive algorithm.
A~call on~$G$ returns an independent set~$I_G$ and an integer~$u_G\geq\alpha(G)$, with~$|I_G|$ approximating~$u_G$.
For each~$k$ (guess of the size of a~maximum independent set), the outputs~$(I_C,u_C)$ from the components of a~balanced separator either yield a~set of size close to~$k$ or a~valid inequality violated by the current query, which we can transmit to the ellipsoid algorithm.
We first establish these inequalities and explain their use in the overall algorithm.

\subsection{Feasibility and valid inequalities}\label{subsec:feasibility}

For a~graph~$G$, an integer~$r\geq1$, and an integer target~$k$, let $d$ be as in~\eqref{eq:dimension}.
We write
\[
  z=(z_A)_{A\subseteq V(G),\,|A|\leq r+2}\in\mathbb R^d
\]
for a~variable vector and~$y\in\mathbb Q^d$ for the query point.
Define
\begin{equation}\label{eq:polytope}
  \mathcal P_r(G,k):=
  \conv\left\{
    \bigl(\mathbf1[A\subseteq I]\bigr)_{\substack{A\subseteq V(G)\\|A|\leq r+2}}
    : I\text{ is an independent set of }G,\ |I|=k
  \right\}.
\end{equation}
Here $\conv$ denotes the convex hull, with $\conv\varnothing=\varnothing$.
Thus
\begin{equation}\label{eq:polytope-properties}
  \mathcal P_r(G,k)\neq\varnothing\ \Longleftrightarrow\ \alpha(G)\geq k,
  \qquad \mathcal P_r(G,k)\subseteq R_r(G,k).
\end{equation}

An inequality is \emph{valid for}~$\mathcal P_r(G,k)$ if every point of this polytope satisfies it.
We already observed that the constraints in~\eqref{eq:relaxation} are valid for~$\mathcal P_r(G,k)$.
The following lemma gives further such inequalities from upper bounds on the independence numbers of induced subgraphs.

\begin{lemma}\label{lem:valid-cut}
  Let $C\subseteq V(G)$ and $u_C\geq\alpha(G[C])$.
  For every $A\subseteq V(G)$ with $|A|\leq r$, every vector $z\in\mathcal P_r(G,k)$ satisfies
  \begin{equation}\label{eq:valid-cut}
    \sum_{v\in C}z_{A\cup\{v\}}\leq u_Cz_A.
  \end{equation}
  A~point $y\in\mathbb R^d$ with $y_A>0$ violates this inequality whenever $\sum_{v\in C}y_{A\cup\{v\}}/y_A>u_C$.
\end{lemma}
\begin{proof}
  Let~$I$ be an independent set of~$G$ of size~$k$, and evaluate~\eqref{eq:valid-cut} at the corresponding vector in~\eqref{eq:polytope}.
  If $A\nsubseteq I$, both sides of~\eqref{eq:valid-cut} are~0.
  If $A\subseteq I$, the inequality reads $|I\cap C|\leq u_C$, which holds since~$I\cap C$ is an independent set of~$G[C]$.
  By linearity, the inequality holds for every convex combination of these vectors, hence throughout~$\mathcal P_r(G,k)$.
  For $y_A>0$, dividing~\eqref{eq:valid-cut} by~$y_A$ shows that it is violated precisely when
  \[
    \sum_{v\in C}\frac{y_{A\cup\{v\}}}{y_A}>u_C.\qedhere
  \]
\end{proof}

A~\emph{strong separation oracle} for a~rational polytope~$Q\subseteq\mathbb R^d$, on input~$y\in\mathbb Q^d$, either certifies $y\in Q$ or returns a~linear inequality valid for~$Q$ and violated by~$y$.
We use the exact rational ellipsoid feasibility algorithm of Gr\"otschel, Lov\'asz, and Schrijver~\cite[Theorem~6.4.1]{Grotschel93}.
Given such an oracle and a~bound on the encoding length of each inequality in a~description of~$Q$, it finds a~point of~$Q$ or certifies emptiness.
This guarantee also applies when the polytope is contained in a~proper affine subspace.
The number of oracle calls and the additional work are polynomial in~$d$ and the encoding bound.
Here $\mathcal P_r(G,k)\subseteq[0,1]^d$ has 0--1 vertices of encoding length~$\Oh(d)$, so a~polynomial bound in~$d$ on the encoding length of each defining inequality follows from~\cite[Lemma~6.2.4]{Grotschel93}.

We run the ellipsoid algorithm for~$\mathcal P_r(G,k)$, with an early stopping rule.
Each query~$y$ is passed to~$\RoundOrCut(G,k,y)$.
This procedure returns either a~sufficiently large (approximate) independent set, in which case we stop, or an inequality valid for~$\mathcal P_r(G,k)$ and violated by~$y$, in which case the ellipsoid algorithm continues.

\Cref{subsec:recursion} proves that the procedure always returns one of these two answers.
Consequently, until an independent set is returned, every answer is a~separation for the same fixed polytope.
The exact feasibility guarantee then gives either an independent set or a~certificate that~$\mathcal P_r(G,k)$ is empty. 
This is the \emph{round-or-cut} approach introduced by Carr et~al.~\cite{CarrFLP00}.
Chakrabarty and Negahbani~\cite[Section~3.2]{Chakrabarty18} similarly use an auxiliary optimization problem to obtain either an approximate solution or an inequality separating the current query from a~fixed polytope.

\subsection{Algorithm and correctness}\label{subsec:recursion}

On an edgeless input, return all vertices.
Otherwise, let~$n\geq2$ be the number of vertices of the original input, and fix rational $0<\varepsilon<1$.
Put
\begin{equation}\label{eq:global-parameters}
  \lambda:=\lceil\log_{3/2}n\rceil,
  \qquad\delta:=\frac{\varepsilon}{\lambda}.
\end{equation}
Choose the integer threshold~$p$ given by~\cref{lem:cheap-separator} for~$n$ and~$\delta$, and set $r:=2p\lambda$.
We use the nondecreasing choice of~$p$ specified after~\cref{lem:no-flow}.
Thus the same threshold applies to every recursive induced subgraph, and $r\geq2p\log n$ satisfies~\cref{lem:conditioning}.
All these parameters remain fixed throughout the recursion, with $p=\Oh_{\mathcal C}\bigl((\lambda/\delta)^{4a-2}\bigr)$.

\subparagraph*{The procedure {\fontfamily{cmr}\bfseries\scshape Solve}.}
For an integer $t\geq1$, put $h(t):=\lceil\log_{3/2}t\rceil$, and set $h(0):=0$.
We define a~procedure~$\Solve(G)$ that returns an independent set~$I_G$ and an integer~$u_G$ satisfying
\begin{equation}\label{eq:invariant}
  \alpha(G)\leq u_G\leq |V(G)|,
  \qquad |I_G|\geq(1-\delta)^{h(|V(G)|)}u_G.
\end{equation}
Each recursive call from~$G$ is on~$G[C]$ with~$|C|\leq2|V(G)|/3$, where~$C$ is a~component left after deleting a~balanced separator.
Thus~$h(|V(G)|)$ bounds the recursion depth.

The procedure considers the targets $k=|V(G)|,|V(G)|-1,\ldots,1$.
For each target, it either returns an independent set of size at least
\[
  (1-\delta)^{h(|V(G)|)}k,
\]
or proves $\alpha(G)<k$.
For $k\leq p$, it searches exhaustively for an independent set of size~$k$; for $k>p$, it runs the ellipsoid algorithm with queries handled by~\hyperref[par:round-or-cut]{$\RoundOrCut$}, described below.
At the first successful target~$k$, the preceding rejections give $\alpha(G)\leq k$, so the procedure returns $u_G:=k$.
A~returned independent set establishes the approximation bound for the current target.
The pseudocode is given in~\cref{alg:solve}.

\medskip
\begingroup
\SetAlFnt{\small}
\SetAlgoSkip{}
\setlength{\intextsep}{6pt plus 2pt minus 2pt}
\begin{algorithm}[!htbp]
  \caption{$\Solve(G)$}\label{alg:solve}
  \KwIn{An induced subgraph~$G$ of the input.}
  \KwData{The shared parameters~$\lambda,\delta,p,r$, fixed for the original input.}
  \KwOut{A~pair~$(I_G,u_G)$ satisfying~\eqref{eq:invariant}, with~$I_G$ an independent set of~$G$.}
  \If{$E(G)=\varnothing$}{
    \Return $(V(G),|V(G)|)$\;
  }
  \For{$k=|V(G)|,|V(G)|-1,\ldots,1$}{
    \eIf{$k\leq p$}{
      \ForEach{$I\subseteq V(G)$ with $|I|=k$}{
        \If{$I$ is independent}{
          \Return $(I,k)$\;
        }
      }
    }{
      Run the exact ellipsoid feasibility algorithm for~$\mathcal P_r(G,k)$\;
      Handle each oracle query or feasible output~$y$ as follows:\;
      \Indp
      Call $\RoundOrCut(G,k,y)$\;
      \eIf{it returns an independent set~$I$}{
        Stop the ellipsoid algorithm and \Return $(I,k)$\;
      }{
        Pass the returned inequality to the ellipsoid algorithm and resume it\;
      }
      \Indm
      Once the ellipsoid algorithm certifies emptiness, continue with the next target\;
    }
  }
\end{algorithm}
\endgroup

\subparagraph*{The procedure {\fontfamily{cmr}\bfseries\scshape RoundOrCut}.}\label{par:round-or-cut}
For a~target~$k>p$, the procedure~$\RoundOrCut(G,k,y)$ first checks~\eqref{eq:relaxation} and returns a~violated inequality if one exists.
When a~query violates an equality in~\eqref{eq:relaxation}, it returns the violated one of its two associated inequalities.
For a~query satisfying~\eqref{eq:relaxation}, it obtains~$(A,x)$ by~\cref{lem:conditioning} and finds a~$(\delta,x)$-light balanced separator~$S$ by~\cref{lem:cheap-separator}.
It then calls~$\Solve(G[C])$ to obtain~$(I_C,u_C)$ for each component~$C$ of~$G-S$.

Each recursive call on~$G[C]$ uses the shared parameters.
The parent computes $x(C)=\sum_{v \in C}y_{A\cup\{v\}}/y_A$ for its query~$y$ and compares it with the returned bound~$u_C$.
If $x(C)>u_C$ for some component, the parent returns~\eqref{eq:valid-cut} to the ellipsoid algorithm for~$\mathcal P_r(G,k)$; otherwise, it returns the union of the sets~$I_C$.
The pseudocode is given in~\cref{alg:round-or-cut}.

\begingroup
\SetAlFnt{\small}
\SetAlgoSkip{}
\setlength{\intextsep}{6pt plus 2pt minus 2pt}
\begin{algorithm}[!htbp]
  \caption{$\RoundOrCut(G,k,y)$}\label{alg:round-or-cut}
  \KwIn{An induced subgraph~$G$ of the input, an integer~$k>p$, and a rational query vector~$y$, with one coordinate $y_A$ for each $A\subseteq V(G)$ with $|A| \leq r+2$.}
  \KwData{The shared parameters~$\lambda,\delta,p,r$, fixed for the original input.}
  \KwOut{An independent set~$I$ with $|I|\geq(1-\delta)^{h(|V(G)|)}k$, or an inequality valid for the polytope~$\mathcal P_r(G,k)$ defined by~\eqref{eq:polytope} and violated by~$y$.}
  \If{$y$ violates~\eqref{eq:relaxation}}{
    \Return a~violated inequality from~\eqref{eq:relaxation}\;
  }
  Obtain $(A,x)$ by~\cref{lem:conditioning} with parameter~$p$\;
  Obtain a~$(\delta,x)$-light balanced separator~$S$ by~\cref{lem:cheap-separator}\;
  $I\leftarrow\varnothing$\;
  \ForEach{component~$C$ of~$G-S$}{
    $(I_C,u_C)\leftarrow\Solve(G[C])$\;
    \If{$x(C)>u_C$}{
      \Return $\displaystyle\sum_{v\in C}z_{A\cup\{v\}}\leq u_Cz_A$\;
    }
    $I\leftarrow I\cup I_C$\;
  }
  \Return $I$\;
\end{algorithm}
\endgroup

\FloatBarrier
\begin{lemma}\label{lem:correctness}
  Every call $\Solve(G)$ terminates and satisfies~\eqref{eq:invariant}.
\end{lemma}
\begin{proof}
  We argue by induction on $n':=|V(G)|$.
  The edgeless case, including the empty graph, is immediate.
  Otherwise, every recursive call is on an induced subgraph with at most $2n'/3 < n'$ vertices, so the induction hypothesis applies to it.

  Consider $k>p$ and a~query~$y$ generated by the ellipsoid algorithm for~$\mathcal P_r(G,k)$.
  If~$y$ violates~\eqref{eq:relaxation}, the returned inequality is valid for~$\mathcal P_r(G,k)$ by~\eqref{eq:polytope-properties}.
  Otherwise, conditioning returns~$A$ with $y_A>0$ and a~vector $x \in[0,1]^{V(G)}$ satisfying $x(V(G))=k$ and $e_G(x)\leq k^2/p$.
  Since $k>p$, \cref{lem:cheap-separator} applies and returns a~balanced separator~$S$ of~$G$ with $x(S)\leq\delta k$.
  For every component~$C$ of~$G-S$, induction gives $u_C\geq\alpha(G[C])$.
  If $x(C)>u_C$, \cref{lem:valid-cut} therefore gives a~valid violated inequality for the fixed polytope~$\mathcal P_r(G,k)$.

  Suppose that every component satisfies $x(C)\leq u_C$.
  The union~$I$ of the sets~$I_C$ is independent since there are no edges between distinct components of~$G-S$.
  Since~$S$ is balanced, every nonempty component satisfies $|C|\leq2n'/3$, and hence $h(|C|)\leq h(n')-1$.
  Thus the induction hypothesis yields
  \begin{align}
    |I|=\sum_C|I_C|
      &\geq(1-\delta)^{h(n')-1}\sum_Cu_C\notag\\
      &\geq(1-\delta)^{h(n')-1}\sum_Cx(C)\notag\\
      &=(1-\delta)^{h(n')-1}(k-x(S))\notag\\
      &\geq(1-\delta)^{h(n')}k.\label{eq:rounding}
  \end{align}
  The last inequality uses the fact that~$S$ is~$(\delta,x)$-light.
  Thus~$\RoundOrCut$ terminates and returns either an independent set of size at least~$(1-\delta)^{h(n')}k$ or a~valid violated inequality.
  By the remarks in~\cref{subsec:feasibility}, for each $k>p$, the ellipsoid algorithm with the early stopping rule terminates with an independent set or a~certificate that $\mathcal P_r(G,k)=\varnothing$.

  For $k\leq p$, we only reject the target~$k$ if exhaustive search finds no independent set of size~$k$, so $\alpha(G)<k$.
  For $k>p$, rejection means that the ellipsoid algorithm certifies $\mathcal P_r(G,k)=\varnothing$, which also implies $\alpha(G)<k$ by~\eqref{eq:polytope-properties}.
  Every rejected target therefore is larger than~$\alpha(G)$.
  Since the computation for each target terminates, the loop must return an independent set by the time it reaches~$k=\alpha(G)$.

  Let~$k$ be the first successful run.
  If $k=|V(G)|$, then $\alpha(G)\leq k$ directly.
  Otherwise, target~$k+1$ was rejected, so $\alpha(G)<k+1$, and integrality gives $\alpha(G)\leq k$.
  Thus returning $u_G:=k$ gives the required upper bound.
  The size guarantee follows from~\eqref{eq:rounding} when $k>p$, and from $|I_G|=k$ otherwise.
\end{proof}

\begin{corollary}\label{cor:unweighted}
  Let~$G\in\mathcal C$ and let~$0<\varepsilon<1$ be rational.
  With the parameter choices above, $\Solve(G)$ returns an independent set~$I_G$ satisfying $|I_G|\geq(1-\varepsilon)\alpha(G)$.
\end{corollary}

\begin{proof}
  The edgeless case is immediate.
  Otherwise, $h(n)=\lambda$, so~\cref{lem:correctness} and Bernoulli's inequality give
  \[
    |I_G|\geq(1-\delta)^{h(n)}u_G
      \geq(1-\delta h(n))\alpha(G)
      =(1-\varepsilon)\alpha(G).\qedhere
  \]
\end{proof}

\FloatBarrier
\subsection{Running time and rational arithmetic}\label{subsec:time}

For the time analysis, we may assume that~$\varepsilon$ has encoding length~$\Oh(\log\varepsilon^{-1})$.
Indeed, a~rational~$\varepsilon$ can be rounded down to a~dyadic number~$\varepsilon'$ with $\varepsilon/2<\varepsilon'\leq\varepsilon$, in polynomial time in its encoding length.
Using~$\varepsilon'$ preserves the requested guarantee and changes the final bound by at most a~constant factor in the exponent.

All constants below may depend on~$\mathcal C$.
The parameter choices in~\cref{subsec:recursion} give
\begin{equation}\label{eq:parameter-bounds}
  p=\Oh_{\mathcal C}(\varepsilon^{-(4a-2)}\lambda^{8a-4}),
  \qquad r=\Oh(p\lambda).
\end{equation}
In every recursive call, the dimension in~\eqref{eq:dimension} and the number of constraints in~\eqref{eq:relaxation} are at most~$n^{\Oh(r)}$, where~$n$ is the original input size.
\Cref{subsec:feasibility} establishes a~polynomial bound on the encoding length of each inequality in some description of~$\mathcal P_r(G,k)$.
Every returned inequality has integer coefficients of absolute value~$\Oh(n)$: this holds directly for~\eqref{eq:relaxation}, and for~\eqref{eq:valid-cut} it follows from $u_C\leq|C|$, after repeated coordinates are collected.
Consequently, over all targets in one call to~$\Solve$, the ellipsoid algorithm makes at most~$n^{\Oh(r)}$ queries and uses~$n^{\Oh(r)}$ additional time.
Its rational queries have encoding length~$n^{\Oh(r)}$.

Processing one query takes~$n^{\Oh(r)}$ time, excluding the time spent in recursive calls to~$\Solve$.
This includes checking~\eqref{eq:relaxation}, enumerating the conditioning sets in~\cref{lem:conditioning}, and applying the separator algorithm of~\cref{lem:cheap-separator}.
The coordinates of~$x$ are ratios of entries of the current query, so their encoding lengths are polynomial in the encoding length of that query.
All comparisons are performed by exact rational arithmetic.

Enumeration for targets $k\leq p$ takes~$n^{\Oh(p)}$ time, which is absorbed in~$n^{\Oh(r)}$.
A~call to~$\Solve$ therefore performs~$n^{\Oh(r)}$ work outside its recursive calls and makes at most~$n^{\Oh(r)}$ recursive calls: there are~$n^{\Oh(r)}$ queries, each generating at most~$n$ child calls.
Each child has at most two thirds as many vertices as its parent, so the recursion depth is at most~$\lambda$.
The total running time is consequently
\begin{equation}\label{eq:running-time}
  n^{\Oh(r\lambda)}
    =2^{\Oh(r\lambda^2)}
    =2^{\Oh_{\mathcal C}(p\lambda^3)}
    =2^{\Oh_{\mathcal C}(\varepsilon^{-(4a-2)}(\log n)^{8a-1})}.
\end{equation}
The dimension and coefficient bounds above apply to every recursive call, so every ellipsoid query has encoding length~$n^{\Oh(r)}$.
For fixed~$\varepsilon$ and~$\mathcal C$, \eqref{eq:running-time} is quasipolynomial.
Together with~\cref{cor:unweighted}, this gives a~deterministic QPTAS for \mis on~$\mathcal C$.

\section{Vertex weights}\label{sec:weights}

We reduce \swmis to \smis by rounding the vertex weights and replacing each vertex by an independent set whose size is determined by its rounded weight.
We first show that this operation preserves strongly edge-sublinear separators, possibly with a~larger exponent.

An \emph{independent blow-up} of~$G$ is obtained by replacing each vertex~$v$ by a~nonempty independent set~$B_v$.
For distinct~$u,v$, all edges between~$B_u$ and~$B_v$ are present if $uv\in E(G)$, and none are present otherwise.
Let~$\mathcal C^{\mathrm{bl}}$ be the class of independent blow-ups of graphs in~$\mathcal C$.

\begin{lemma}\label{lem:blow-ups}
  Suppose~$\mathcal C$ is hereditary and satisfies~\eqref{eq:treewidth}.
  Then~$\mathcal C^{\mathrm{bl}}$ is hereditary, and every $m$-edge graph~$G'\in\mathcal C^{\mathrm{bl}}$ satisfies
  \[
    \tw(G')=\Oh_{\mathcal C}\bigl(m^{1-1/(2a)}\bigr).
  \]
  Consequently, $\mathcal C^{\mathrm{bl}}$ has strongly edge-sublinear separators.
\end{lemma}
\begin{proof}
  The class~$\mathcal C^{\mathrm{bl}}$ is hereditary: an induced subgraph of an independent blow-up of~$G$ is an independent blow-up of~$G[U]$, where~$U$ consists of the vertices of~$G$ with at least one surviving copy.

  Let~$G'$ be an independent blow-up of~$G\in\mathcal C$, and assume $m:=|E(G')|\geq1$.
  For an integer $s\geq1$, set
  \[
    S:=\{v\in V(G'):|N_{G'}(v)|>s\}.
  \]
  By the degree sum formula, $|S|\leq2m/s$.
  Consider a~set~$B_v$ with more than~$s$ vertices.
  Every neighbor of~$B_v$ has degree at least $|B_v|>s$, and hence belongs to~$S$.
  Thus all vertices of~$B_v\setminus S$ are isolated in~$G'-S$.

  It follows that, apart from isolated vertices, $G'-S$ is obtained from~$G$ by replacing each vertex by a~possibly empty independent set of at most~$s$ vertices.
  A~tree-decomposition of~$G$ can therefore be transformed into one of this graph by replacing each vertex in every bag by its copies that are nonisolated in~$G'-S$.
  The bag sizes increase by a~factor of at most~$s$.
  Adding the isolated vertices in singleton bags, and then adding~$S$ to every bag, gives
  \[
    \tw(G')+1\leq\frac{2m}{s}+s\bigl(\tw(G)+1\bigr).
  \]
  Since $|E(G)|\leq m$, \eqref{eq:treewidth} gives $\tw(G)+1=\Oh_{\mathcal C}(m^{1-1/a})$.
  Choosing $s:=\lceil m^{1/(2a)}\rceil$ yields
  \[
    \tw(G')+1
      =\Oh_{\mathcal C}\left(\frac{m}{s}+s m^{1-1/a}\right)
      =\Oh_{\mathcal C}\bigl(m^{1-1/(2a)}\bigr).
  \]
  For $m=0$, the graph~$G'$ is edgeless.
  Finally, a~tree-decomposition of width~$w$ gives a~balanced separator of size at most~$w+1$.
\end{proof}

\mainresult*

\begin{proof}
  Let~$G\in\mathcal C$ have~$n$ vertices with nonnegative rational weights~$w_v$, and fix rational $0<\varepsilon<1$.
  The cases $n\leq1$ and $w_v=0$ for all~$v$ are immediate.
  Otherwise, put
  \[
    w_{\max}:=\max_{v\in V(G)}w_v,
    \qquad \tau:=\frac{\varepsilon w_{\max}}{2n},
    \qquad t_v:=\left\lfloor\frac{w_v}{\tau}\right\rfloor.
  \]
  Construct~$G'$ by replacing each vertex~$v$ by an independent set of size~$t_v$, omitting vertices with $t_v=0$.
  Then
  \[
    |V(G')|=\sum_{v\in V(G)}t_v
      \leq\sum_{v\in V(G)}\frac{w_v}{\tau}
      \leq\frac{nw_{\max}}{\tau}
      =\frac{2n^2}{\varepsilon}.
  \]
  Since~$\mathcal C$ is hereditary, $G'\in\mathcal C^{\mathrm{bl}}$.
  By~\cref{lem:blow-ups}, we may apply the unweighted QPTAS from~\cref{sec:algorithm} to~$G'$ with accuracy~$\varepsilon/2$, obtaining an independent set~$I'$.
  Return the set~$I$ of vertices of~$G$ having at least one copy in~$I'$; $I$ is an independent set.

  Write $w(J):=\sum_{v\in J}w_v$, and let~$I^*$ be an optimum weighted independent set of~$G$, of weight $\alpha_w(G):=w(I^*)$.
  Since all copies of the vertices in~$I^*$ form an independent set in~$G'$, we have
  \[
    \tau\alpha(G')\geq\tau\sum_{v\in I^*}t_v
      \geq\sum_{v\in I^*}(w_v-\tau)
      =\alpha_w(G)-|I^*|\tau
      \geq\alpha_w(G)-n\tau
      \geq(1-\varepsilon/2)\alpha_w(G),
  \]
  where the last inequality uses $n\tau=\varepsilon w_{\max}/2\leq\varepsilon\alpha_w(G)/2$.
  Since $w_v\geq\tau t_v$ and~$I'$ contains at most~$t_v$ copies of each~$v\in I$, we have $w(I)\geq\tau|I'|$.
  Thus
  \[
    w(I)\geq(1-\varepsilon/2)\tau\alpha(G')
      \geq(1-\varepsilon/2)^2\alpha_w(G)
      \geq(1-\varepsilon)\alpha_w(G).
  \]

  The integers~$t_v$ are computed by exact rational arithmetic, in time polynomial in the binary encoding length of the weights and~$\varepsilon$.
  The number of copies is polynomial in~$n$ and~$1/\varepsilon$.
  By~\cref{lem:blow-ups}, the parameter~$a$ in~\eqref{eq:running-time} is replaced by~$2a$ for~$\mathcal C^{\mathrm{bl}}$.
  Consequently, the unweighted call takes time
  \[
    2^{\Oh_{\mathcal C}\left(\varepsilon^{-(8a-2)}
       (\log n+\log\varepsilon^{-1})^{16a-1}\right)}.
  \]
  Since $n\geq2$, this is $2^{(\varepsilon^{-1}\log n)^{\Oh_{\mathcal C}(1)}}$.
  Including the polynomial work on the original weights gives the stated running time.
\end{proof}

\subparagraph*{AI disclosure.}
GPT-6 Pro first found a randomized QPTAS for \smis in any $H$-induced-minor-free class.
The initial argument relied on the exponent not exceeding $1/2$ in the balanced separators of size $\widetilde{O}(m^{1/2})$ for $m$-edge graphs.
We thought that this should only be an artifact of the proof, and suggested the new notion of strongly edge-sublinear separators.
Our intuition turned out to be correct.
Making the algorithm deterministic and able to handle weights was relatively simple.
The authors rewrote the eventual proof.

This is a typical example where one needs to be well familiar with a~selection of techniques, here: (1) the Leighton--Rao multicommodity flow--cut theorem and its use in the context of excluded induced minors following Korhonen and Lokshtanov~\cite{Korhonen24}, (2) conditioning and entropy-decrease arguments, (3) the round-or-cut approach.

\noindent An autoformalized Lean proof can be found on Lax: \url{https://laxarchive.org/lax-263307}.

\FloatBarrier

\begingroup
\interlinepenalty=10000
\bibliography{round-or-cut}

@article{Baker94,
  author       = {Brenda S. Baker},
  title        = {Approximation Algorithms for {NP}-Complete Problems on Planar Graphs},
  journal      = {J. {ACM}},
  volume       = {41},
  number       = {1},
  pages        = {153--180},
  year         = {1994},
  url          = {https://doi.org/10.1145/174644.174650},
  doi          = {10.1145/174644.174650},
  bibsource    = {dblp computer science bibliography, https://dblp.org}
}

@inproceedings{Korhonen24,
  author       = {Tuukka Korhonen and
                  Daniel Lokshtanov},
  editor       = {David P. Woodruff},
  title        = {Induced-Minor-Free Graphs: Separator Theorem, Subexponential Algorithms,
                  and Improved Hardness of Recognition},
  booktitle    = {Proceedings of the 2024 Annual {ACM-SIAM} Symposium on Discrete
                  Algorithms ({SODA} 2024)},
  pages        = {5249--5275},
  publisher    = {{SIAM}},
  year         = {2024},
  url          = {https://doi.org/10.1137/1.9781611977912.188},
  doi          = {10.1137/1.9781611977912.188},
  bibsource    = {dblp computer science bibliography, https://dblp.org},
  eprint       = {2308.04795},
  archivePrefix = {arXiv}
}

@inproceedings{Karp72,
  author       = {Richard M. Karp},
  editor       = {Raymond E. Miller and
                  James W. Thatcher},
  title        = {Reducibility Among Combinatorial Problems},
  series       = {The {IBM} Research Symposia Series},
  pages        = {85--103},
  publisher    = {Plenum Press, New York},
  year         = {1972},
  url          = {https://doi.org/10.1007/978-1-4684-2001-2\_9},
  doi          = {10.1007/978-1-4684-2001-2_9},
  bibsource    = {dblp computer science bibliography, https://dblp.org}
}

@article{Hastad99,
  author = {H{\aa}stad, Johan},
  title = {Clique Is Hard to Approximate within {$n^{1-\varepsilon}$}},
  journal = {Acta Mathematica},
  volume = {182},
  number = {1},
  pages = {105--142},
  year = {1999},
  doi = {10.1007/BF02392825}
}

@article{Zuckerman07,
  author = {Zuckerman, David},
  title = {Linear Degree Extractors and the Inapproximability of {Max Clique} and {Chromatic Number}},
  journal = {Theory of Computing},
  volume = {3},
  number = {6},
  pages = {103--128},
  year = {2007},
  doi = {10.4086/toc.2007.v003a006}
}

@inproceedings{Lee17,
  author       = {James R. Lee},
  editor       = {Christos H. Papadimitriou},
  title        = {Separators in Region Intersection Graphs},
  booktitle    = {8th Innovations in Theoretical Computer Science Conference, {ITCS}
                  2017, January 9-11, 2017, Berkeley, CA, {USA}},
  series       = {LIPIcs},
  volume       = {67},
  pages        = {1:1--1:8},
  publisher    = {Schloss Dagstuhl - Leibniz-Zentrum f{\"{u}}r Informatik},
  year         = {2017},
  url          = {https://doi.org/10.4230/LIPIcs.ITCS.2017.1},
  doi          = {10.4230/LIPICS.ITCS.2017.1},
  bibsource    = {dblp computer science bibliography, https://dblp.org}
}

@article{Dallard24c,
  author = {Cl{\'{e}}ment Dallard and
                  Martin Milani{\v{c}} and
                  Kenny {\v{S}}torgel},
  title = {Treewidth versus clique number. {III.} {T}ree-independence number of
                  graphs with a forbidden structure},
  journal = {Journal of Combinatorial Theory, Series B},
  volume = {167},
  pages = {338--391},
  year = {2024},
  doi = {10.1016/j.jctb.2024.03.005},
  eprint = {2206.15092},
  archivePrefix = {arXiv}
}

@phdthesis{Gartland23,
  author = {Peter Gartland},
  title = {Quasi-Polynomial Time Techniques for Independent Set and Beyond in
                  Hereditary Graph Classes},
  school = {University of California, Santa Barbara},
  year = {2023},
  url = {https://www.escholarship.org/uc/item/0kk6d2jv}
}

@article{Grzesik22,
  author = {Andrzej Grzesik and Tereza Klimo{\v s}ov{\'a} and Marcin Pilipczuk and Micha{\l} Pilipczuk},
  title = {Polynomial-Time Algorithm for Maximum Weight Independent Set on {$P_6$}-Free Graphs},
  journal = {ACM Trans. Algorithms},
  volume = {18},
  number = {1},
  pages = {4:1--4:57},
  year = {2022},
  doi = {10.1145/3414473},
  url = {https://doi.org/10.1145/3414473}
}

@inproceedings{Gartland20,
  author = {Gartland, Peter and Lokshtanov, Daniel},
  title = {Independent Set on {$P_k$}-Free Graphs in Quasi-Polynomial Time},
  booktitle = {2020 IEEE 61st Annual Symposium on Foundations of Computer Science (FOCS)},
  pages = {613--624},
  publisher = {IEEE},
  year = {2020},
  doi = {10.1109/FOCS46700.2020.00063}
}

@inproceedings{Gartland21,
  author = {Peter Gartland and Daniel Lokshtanov and Marcin Pilipczuk and Micha{\l} Pilipczuk and Pawe{\l} Rz{\k{a}}{\.z}ewski},
  title = {Finding large induced sparse subgraphs in {$C_{>t}$}-free graphs in quasipolynomial time},
  booktitle = {STOC},
  year = {2021},
  eprint = {2007.11402v3},
  archivePrefix = {arXiv},
  doi = {10.1145/3406325.3451034},
  url = {https://doi.org/10.1145/3406325.3451034}
}

@article{Robertson86a,
  author  = {Neil Robertson and Paul D. Seymour},
  title   = {Graph Minors. {II}. {A}lgorithmic Aspects of Tree-Width},
  journal = {Journal of Algorithms},
  volume  = {7},
  number  = {3},
  pages   = {309--322},
  year    = {1986},
  doi     = {10.1016/0196-6774(86)90023-4}
}

@article{Dvorak19,
  author  = {Zden{\v{e}}k Dvo{\v{r}}{\'a}k and Sergey Norin},
  title   = {Treewidth of Graphs with Balanced Separations},
  journal = {Journal of Combinatorial Theory, Series B},
  volume  = {137},
  pages   = {137--144},
  year    = {2019},
  doi     = {10.1016/j.jctb.2018.12.007}
}

@inproceedings{Dinur14,
  author    = {Irit Dinur and David Steurer},
  title     = {Analytical Approach to Parallel Repetition},
  booktitle = {Proceedings of the Forty-Sixth Annual ACM Symposium on Theory of Computing (STOC 2014)},
  pages     = {624--633},
  publisher = {ACM},
  year      = {2014},
  doi       = {10.1145/2591796.2591884}
}

@article{Leighton99,
  author = {Leighton, Tom and Rao, Satish},
  title = {Multicommodity Max-Flow Min-Cut Theorems and Their Use in Designing Approximation Algorithms},
  journal = {Journal of the ACM},
  volume = {46},
  number = {6},
  pages = {787--832},
  year = {1999},
  doi = {10.1145/331524.331526}
}

@book{Grotschel93,
  author = {Gr{\"o}tschel, Martin and Lov{\'a}sz, L{\'a}szl{\'o} and Schrijver, Alexander},
  title = {Geometric Algorithms and Combinatorial Optimization},
  edition = {Second},
  series = {Algorithms and Combinatorics},
  volume = {2},
  publisher = {Springer},
  year = {1993},
  doi = {10.1007/978-3-642-78240-4}
}

@article{Hoory06,
  author = {Hoory, Shlomo and Linial, Nathan and Wigderson, Avi},
  title = {Expander Graphs and Their Applications},
  journal = {Bulletin of the American Mathematical Society},
  volume = {43},
  number = {4},
  pages = {439--561},
  year = {2006},
  doi = {10.1090/S0273-0979-06-01126-8}
}

@article{Grohe09,
  author = {Martin Grohe and D{\'a}niel Marx},
  title = {On Tree Width, Bramble Size, and Expansion},
  journal = {Journal of Combinatorial Theory, Series B},
  volume = {99},
  number = {1},
  pages = {218--228},
  year = {2009},
  doi = {10.1016/j.jctb.2008.06.004}
}

@article{Sherali90,
  author = {Sherali, Hanif D. and Adams, Warren P.},
  title = {A Hierarchy of Relaxations between the Continuous and Convex Hull Representations for Zero-One Programming Problems},
  journal = {SIAM Journal on Discrete Mathematics},
  volume = {3},
  number = {3},
  pages = {411--430},
  year = {1990},
  doi = {10.1137/0403036}
}

@inproceedings{Barak11,
  author = {Barak, Boaz and Raghavendra, Prasad and Steurer, David},
  title = {Rounding Semidefinite Programming Hierarchies via Global Correlation},
  booktitle = {52nd Annual IEEE Symposium on Foundations of Computer Science (FOCS 2011)},
  pages = {472--481},
  publisher = {IEEE},
  year = {2011},
  doi = {10.1109/FOCS.2011.95}
}

@inproceedings{RT12,
  author = {Raghavendra, Prasad and Tan, Ning},
  title = {Approximating {CSPs} with Global Cardinality Constraints Using {SDP} Hierarchies},
  booktitle = {Proceedings of the Twenty-Third Annual ACM-SIAM Symposium on Discrete Algorithms (SODA 2012)},
  pages = {373--387},
  publisher = {SIAM},
  year = {2012},
  doi = {10.1137/1.9781611973099.33}
}

@inproceedings{CarrFLP00,
  author    = {Robert D. Carr and Lisa K. Fleischer and
               Vitus J. Leung and Cynthia A. Phillips},
  title     = {Strengthening Integrality Gaps for Capacitated
               Network Design and Covering Problems},
  booktitle = {Proceedings of the 11th Annual ACM-SIAM
               Symposium on Discrete Algorithms (SODA)},
  pages     = {106--115},
  year      = {2000},
  publisher = {SIAM},
  doi       = {10.1145/338219.338241}
}

@inproceedings{Chakrabarty18,
  author = {Chakrabarty, Deeparnab and Negahbani, Maryam},
  title = {Generalized Center Problems with Outliers},
  booktitle = {45th International Colloquium on Automata, Languages, and Programming (ICALP 2018)},
  series = {LIPIcs},
  volume = {107},
  pages = {30:1--30:14},
  publisher = {Schloss Dagstuhl -- Leibniz-Zentrum f{\"u}r Informatik},
  year = {2018},
  doi = {10.4230/LIPIcs.ICALP.2018.30}
}

@misc{Bonnet26b,
  author = {Bonnet, {\'E}douard and Le, Hung and Pilipczuk, Marcin and Pilipczuk, Micha{\l}},
  title = {Coarse Balanced Separators in Fat-Minor-Free Graphs},
  year = {2026},
  eprint = {2604.11318},
  archivePrefix = {arXiv},
}

@misc{Bonnet26c,
  author = {{\'{E}}douard Bonnet and Yeonsu Chang},
  title = {Max Independent Set Remains {NP}-hard when Excluding a Planar Induced Minor},
  year = {2026},
  eprint = {2609.11285},
  archivePrefix = {arXiv},
  primaryClass = {cs.CC},
}

@inproceedings{Adamaszek13,
  author = {Adamaszek, Anna and Wiese, Andreas},
  title = {Approximation Schemes for Maximum Weight Independent Set of Rectangles},
  booktitle = {54th Annual IEEE Symposium on Foundations of Computer Science (FOCS 2013)},
  year = {2013},
  pages = {400--409},
  publisher = {IEEE},
  doi = {10.1109/FOCS.2013.50}
}

@article{Adamaszek19,
  author = {Adamaszek, Anna and Har-Peled, Sariel and Wiese, Andreas},
  title = {Approximation Schemes for Independent Set and Sparse Subsets of Polygons},
  journal = {Journal of the ACM},
  year = {2019},
  volume = {66},
  number = {4},
  pages = {29:1--29:40},
  doi = {10.1145/3326122}
}

@article{Pilipczuk18,
  author = {Pilipczuk, Micha{\l} and van Leeuwen, Erik Jan and Wiese, Andreas},
  title = {Quasi-Polynomial Time Approximation Schemes for Packing and Covering Problems in Planar Graphs},
  journal = {Algorithmica},
  volume = {82},
  number = {6},
  pages = {1703--1739},
  year = {2020},
  doi = {10.1007/s00453-019-00670-w}
}

@article{Dvorak18,
  author = {Dvo{\v r}{\'a}k, Zden{\v e}k},
  title = {On Classes of Graphs with Strongly Sublinear Separators},
  journal = {European Journal of Combinatorics},
  volume = {71},
  pages = {1--11},
  year = {2018},
  doi = {10.1016/j.ejc.2018.02.032}
}

@inproceedings{Dvorak21,
  author = {Dvo{\v r}{\'a}k, Zden{\v e}k and Lahiri, Abhiruk},
  title = {Approximation Schemes for Bounded Distance Problems on Fractionally Treewidth-Fragile Graphs},
  booktitle = {29th Annual European Symposium on Algorithms (ESA 2021)},
  series = {Leibniz International Proceedings in Informatics (LIPIcs)},
  volume = {204},
  pages = {40:1--40:10},
  year = {2021},
  publisher = {Schloss Dagstuhl -- Leibniz-Zentrum f{\"u}r Informatik},
  doi = {10.4230/LIPIcs.ESA.2021.40}
}

@article{Har-Peled17,
  author = {Har-Peled, Sariel and Quanrud, Kent},
  title = {Approximation Algorithms for Polynomial-Expansion and Low-Density Graphs},
  journal = {SIAM Journal on Computing},
  volume = {46},
  number = {6},
  pages = {1712--1744},
  year = {2017},
  doi = {10.1137/16M1079336}
}

@misc{Har-Peled23,
  author = {Har-Peled, Sariel},
  title = {Approximately: Independence Implies Vertex Cover},
  year = {2023},
  eprint = {2308.00840},
  archivePrefix = {arXiv}
}

@inproceedings{Bonnet26d,
  author = {Bonnet, {\'E}douard and Czy{\.z}ewska, Jadwiga and Masa{\v r}{\'\i}k, Tom{\'a}{\v s} and Pilipczuk, Marcin and Rz{\k a}{\.z}ewski, Pawe{\l}},
  title = {{QPTAS} for {MWIS} and Finding Large Sparse Induced Subgraphs in Graphs with Few Independent Long Holes},
  booktitle = {20th Scandinavian Symposium on Algorithm Theory (SWAT 2026)},
  series = {Leibniz International Proceedings in Informatics (LIPIcs)},
  volume = {370},
  pages = {9:1--9:14},
  year = {2026},
  publisher = {Schloss Dagstuhl -- Leibniz-Zentrum f{\"u}r Informatik},
  doi = {10.4230/LIPIcs.SWAT.2026.9}
}

@article{Dvorak16,
  author = {Dvo{\v r}{\'a}k, Zden{\v e}k},
  title = {Sublinear Separators, Fragility and Subexponential Expansion},
  journal = {European Journal of Combinatorics},
  volume = {52},
  pages = {103--119},
  year = {2016},
  doi = {10.1016/j.ejc.2015.09.001}
}

@article{BiswalLeeRao2010,
  author = {Punyashloka Biswal and James R. Lee and Satish Rao},
  title = {Eigenvalue Bounds, Spectral Partitioning, and Metrical Deformations via Flows},
  journal = {Journal of the ACM},
  volume = {57},
  number = {3},
  pages = {13:1--13:23},
  year = {2010},
  doi = {10.1145/1706591.1706593},
  eprint = {0808.0148},
  archivePrefix = {arXiv}
}

@article{Matousek14,
  author = {Matou{\v s}ek, Ji{\v r}{\'\i}},
  title = {Near-Optimal Separators in String Graphs},
  journal = {Combinatorics, Probability and Computing},
  volume = {23},
  number = {1},
  pages = {135--139},
  year = {2014},
  doi = {10.1017/S0963548313000400}
}

@article{Chudnovsky24,
  author = {Maria Chudnovsky and Marcin Pilipczuk and Micha{\l} Pilipczuk and St{\'e}phan Thomass{\'e}},
  title = {Quasi-Polynomial Time Approximation Schemes for the Maximum Weight Independent Set Problem in {$H$}-Free Graphs},
  journal = {SIAM Journal on Computing},
  volume = {53},
  number = {1},
  pages = {47--86},
  year = {2024},
  doi = {10.1137/20M1333778},
  eprint = {1907.04585},
  archivePrefix = {arXiv}
}

@misc{Chudnovsky25,
  author = {Maria Chudnovsky and J. Pascal Gollin and Matja{\v z} Krnc and Martin Milani{\v c}},
  title = {Dominated Balanced Separators in Wheel-Induced-Minor-Free Graphs},
  year = {2025},
  eprint = {2512.12329},
  archivePrefix = {arXiv}
}

@misc{Czyzewska26,
  author = {Jadwiga Czy{\.z}ewska and Tom{\'a}{\v s} Masa{\v r}{\'\i}k and Marcin Pilipczuk and Amadeus Reinald and Pawe{\l} Rz{\k a}{\.z}ewski},
  title = {Induced {Erd{\H o}s--P{\'o}sa} Property for Long Holes, Long Thetas, and Beyond},
  year = {2026},
  eprint = {2607.07697},
  archivePrefix = {arXiv}
}

@misc{Nikabadi26,
  author = {Nikabadi, Amir and Rz{\k a}{\.z}ewski, Pawe{\l}},
  title = {Induced Packing Treewidth},
  year = {2026},
  eprint = {2607.07595},
  archivePrefix = {arXiv},
  primaryClass = {math.CO}
}

@article{Bonnet25,
  author = {Bonnet, {\'E}douard and Hickingbotham, Robert},
  title = {Induced Minors and Region Intersection Graphs},
  journal = {Innovations in Graph Theory},
  volume = {2},
  pages = {313--327},
  year = {2025},
  doi = {10.5802/igt.15}
}

@article{Arnborg89,
  author = {Arnborg, Stefan and Proskurowski, Andrzej},
  title = {Linear Time Algorithms for {NP}-Hard Problems Restricted to Partial {$k$}-Trees},
  journal = {Discrete Applied Mathematics},
  volume = {23},
  number = {1},
  pages = {11--24},
  year = {1989},
  doi = {10.1016/0166-218X(89)90031-0}
}

@article{Courcelle00,
  author = {Courcelle, Bruno and Makowsky, Johann A. and Rotics, Udi},
  title = {Linear Time Solvable Optimization Problems on Graphs of Bounded Clique-Width},
  journal = {Theory of Computing Systems},
  volume = {33},
  number = {2},
  pages = {125--150},
  year = {2000},
  doi = {10.1007/s002249910009}
}

@article{Oum06,
  author = {Oum, Sang-il and Seymour, Paul D.},
  title = {Approximating Clique-Width and Branch-Width},
  journal = {Journal of Combinatorial Theory, Series B},
  volume = {96},
  number = {4},
  pages = {514--528},
  year = {2006},
  doi = {10.1016/j.jctb.2005.10.006}
}

@article{Grohe03,
  author = {Grohe, Martin},
  title = {Local Tree-Width, Excluded Minors, and Approximation Algorithms},
  journal = {Combinatorica},
  volume = {23},
  number = {4},
  pages = {613--632},
  year = {2003},
  doi = {10.1007/s00493-003-0037-9}
}

@article{Gavril73,
  author = {Gavril, F{\u a}nic{\u a}},
  title = {Algorithms for a Maximum Clique and a Maximum Independent Set of a Circle Graph},
  journal = {Networks},
  volume = {3},
  number = {3},
  pages = {261--273},
  year = {1973},
  doi = {10.1002/net.3230030305}
}

@article{Chan12,
  author = {Chan, Timothy M. and Har-Peled, Sariel},
  title = {Approximation Algorithms for Maximum Independent Set of Pseudo-Disks},
  journal = {Discrete \& Computational Geometry},
  volume = {48},
  pages = {373--392},
  year = {2012},
  doi = {10.1007/s00454-012-9417-5}
}

@article{Erlebach05,
  author = {Erlebach, Thomas and Jansen, Klaus and Seidel, Eike},
  title = {Polynomial-Time Approximation Schemes for Geometric Intersection Graphs},
  journal = {SIAM Journal on Computing},
  volume = {34},
  number = {6},
  pages = {1302--1323},
  year = {2005},
  doi = {10.1137/S0097539702402676}
}

@article{Chlebik07Boxes,
  author = {Miroslav Chleb{\'i}k and Janka Chleb{\'i}kov{\'a}},
  title = {The Complexity of Combinatorial Optimization Problems on {$d$}-Dimensional Boxes},
  journal = {SIAM Journal on Discrete Mathematics},
  volume = {21},
  number = {1},
  pages = {158--169},
  year = {2007},
  doi = {10.1137/050629276}
}
\endgroup

\end{document}